\documentclass[preprint,12pt]{elsarticle}

\usepackage{amsmath,amssymb,amsthm,mathtools}
\usepackage{bm}
\usepackage{graphicx}
\usepackage{booktabs}
\usepackage{hyperref}

\mathtoolsset{showonlyrefs}

\newtheorem{theorem}{Theorem}[section]
\newtheorem{proposition}[theorem]{Proposition}
\newtheorem{lemma}[theorem]{Lemma}
\newtheorem{corollary}[theorem]{Corollary}
\newtheorem{assumption}[theorem]{Assumption}
\theoremstyle{definition}

\theoremstyle{remark}
\newtheorem{remark}[theorem]{Remark}

\newcommand{\R}{\mathbb R}
\newcommand{\Om}{\Omega}

\newcommand{\eps}{\varepsilon}
\newcommand{\sig}{\sigma}
\newcommand{\lamR}{\lambda_R}
\newcommand{\divg}{\operatorname{div}}

\newcommand{\sgn}{\operatorname{sgn}}
\newcommand{\norm}[1]{\left\lVert #1 \right\rVert}
\newcommand{\abs}[1]{\left\lvert #1 \right\rvert}
\newcommand{\dd}{\,\mathrm d}

\begin{document}

\begin{frontmatter}

\title{Matched asymptotics of Rayleigh-wave fields near cuspidal ridges and gorges}

\author[inst1]{O.M. Kiselev}
\address[inst1]{Innopolis university, Innopolis, Russia \\  o.kiselev@innopolis.ru}

\begin{abstract}
We construct a local matched-asymptotic description of time-harmonic elastic fields generated by Rayleigh waves near cuspidal elements of a traction-free surface. The free surface is represented locally by a cusp graph with exponent $0<\alpha<1$, or equivalently by a vanishing-width horn $b(s)=B s^m$, $m=1/\alpha>1$.
A cuspidal gorge is a zero-opening re-entrant notch: its leading field is the Williams crack-tip field, and the stresses behave as $r^{-1/2}$. The cusp exponent affects the gorge through lower-order corrections and through the stress cut-off produced by rounding the bottom. In contrast, a cuspidal ridge behaves as an elastic horn with vanishing width. The leading admissible free-tip field is asymptotically rigid (bounded stress), distinct from the high-energy branch supported by a finite tip truncation, where stresses grow as $\sigma \sim \rho^{-m}$.  Finite-element calculations for the local static Lam\'e problems support these predictions: the free-tip ridge test confirms the absence of crack-like growth, the truncated ridge recovers the high-energy law, and the gorge stress slope is found to be close to $-1/2$.
\end{abstract}


\end{frontmatter}

\section{Introduction}

Rayleigh waves are controlled by the traction-free boundary.
This makes them useful for near-surface and crustal imaging, but it also makes them sensitive to topography.
Ignoring the free-surface geometry can introduce phase, amplitude and velocity artefacts in surface-wave tomography and full-waveform inversion \cite{Nuber2016,Li2019,Du2024}.
Steep topography can also scatter Rayleigh waves and contaminate receiver-function images \cite{Lu2022}.
Numerical studies of diffraction by mountains, cavities and irregular surfaces show that topographic features may generate diffracted and creeping wavefields \cite{Luzon1997}.

Most wave-propagation studies treat topography as a finite-frequency scattering object: a hill, a canyon, a basin edge or a rough surface.
That point of view is necessary for data modelling.
The question addressed here is more local:
\emph{what asymptotic class does the elastic field belong to near a cuspidal element of a traction-free boundary?}
Two cases must be separated.
A cuspidal ridge is a material horn with vanishing cross-section.
A cuspidal gorge is a re-entrant notch with zero opening angle.
Although these two geometries may be represented by opposite signs of a similar cusp graph, they are not equivalent as local elastic domains.

The ridge case also separates a cusp from an ordinary wedge or cone.
The basic mechanism is energy-flux concentration.
If a Rayleigh-wave packet is guided by a narrow ridge which is slowly varying
on the wavelength scale, then
\[
    \mathcal P(s)
    \simeq
    A_{\rm sec}(s)c_g\mathcal E(s)
    \simeq
    A_{\rm sec}(s)c_g\rho_0 |\dot{\bm U}(s)|^2
    \simeq {\rm const}.
\]
Thus
\[
    |\dot{\bm U}(s)|\asymp A_{\rm sec}(s)^{-1/2}.
\]
For a cusp $A_{\rm sec}(s)\asymp s^m$, $m>1$, this predicts stronger
amplification than in a wedge or cone.  The argument is only a guide: it breaks
down when the local geometric scale reaches the Rayleigh wavelength and it does
not decide whether the endpoint can transmit a finite force or moment
resultant.  If such a resultant is imposed on a truncated section, then
\[
    \sigma_{\rm mean}(s)
    \asymp
    A_{\rm sec}(s)^{-1}
    \asymp s^{-m}.
\]
For $m>1$ this channel has the power-law energy concentration
$E(\delta)\asymp\delta^{1-m}$.  The asymptotic analysis below separates this
high-energy tip-load branch from the admissible free-tip branch.

The paper constructs the local part of a Rayleigh-wave scattering theory:
the outer Rayleigh field supplies the incident loading, the inner cusp problem
selects the admissible singular structures, and matching fixes their
amplitudes.  The emphasis is on formal matched asymptotics and computable
scaling laws, not on a full spectral theory of cuspidal elastic waveguides.

The main conclusion is asymmetric.
For a free cuspidal ridge, the non-destructive limit of truncated horns forces the tip force and moment resultants to vanish; a non-zero resultant requires a finite cut-off of the tip.
Consequently no universal Williams-type stress blow-up is produced by the leading local free-tip field.
For a cuspidal gorge, the local limit is the crack-tip problem and the stress has the leading order $r^{-1/2}$.
This distinction gives a practical diagnostic for numerical simulations of Rayleigh waves near sharp topographic elements.
The finite-strength interpretation is consistent with the geomorphological
idea of threshold hillslopes and limits to relief: steep mountain topography is
controlled by rock-mass strength and landsliding rather than by an arbitrary
geometric continuation of slopes to a singular point
\cite{SchmidtMontgomery1995,Burbank1996ThresholdHillslopes,Montgomery2001ThresholdHillslopes,LarsenMontgomery2012LandslideErosion}.

\section{Relation to previous work}

Topographic effects on Rayleigh waves are well documented in boundary-integral,
finite-difference and full-waveform studies
\cite{BouchonBarker1996,Luzon1997,Robertsson1996,Pan2018}, and they are also
important in surface-wave inversion \cite{Nuber2016,Li2019,Du2024}.  The
formation of Rayleigh waves at a traction-free boundary is classical
\cite{AkiRichards2002}; here the issue is how a cuspidal boundary point changes
the local asymptotic class.

The local singularity background is also substantial.  The Williams expansion
gives the classical $r^{-1/2}$ law for angular corners and cracks
\cite{Williams1952}, while sharp-defect, cusp-inclusion and beak-tip
asymptotics in elasticity were studied in
\cite{MovchanNazarov1992SharpDefects,NazarovPolyakova1993BeakTip}.  More
generally, weighted expansions and intensity factors in singular domains are
standard tools \cite{KozlovMazyaRossmann1997,MazyaNazarovPlamenevskij2000VolI,MazyaNazarovPlamenevskij2000VolII}.
For elastic cusps, the work of Nazarov and coauthors is especially relevant:
peak-shaped and beak-shaped cusps may exhibit essential spectrum, cusp waves,
radiation conditions at the tip, and weighted Korn inequalities
\cite{Nazarov2008SpikedElasticity,BakharevNazarov2009SuperSharpSpike,Nazarov2009EssentialSpectrumCusp,CardoneNazarovTaskinen2009BeakElastic,Nazarov2012WeightedKorn,KozlovNazarov2016ElasticSolidCusps,KozlovNazarov2018ElasticCusp}.
Related radiation conditions for water waves and cusp asymptotics for
parabolic and Laplace--Beltrami problems appear in
\cite{NazarovTaskinen2011WaterWavesCusp,AntonioukKiselevTarkhanov2015,KiselevShestakov2010,GlebovKiselevTarkhanov2018VolII}.

The present paper uses this asymptotic background but keeps a narrower target:
the local Rayleigh-wave loading of two two-dimensional topographic cusps.  It
does not claim a new general theory of elastic cusps or a replacement for the
weighted Korn and spectral results cited above.  Its contribution is the
Rayleigh-wave topographic classification: a cuspidal ridge selects a horn-type
zero-tip-load class unless a finite truncation carries a resultant, whereas a
cuspidal gorge selects the Williams crack-tip class.  This distinction is also
consistent with threshold-slope geomorphology, where steep mountain relief is
limited by rock-mass strength and landsliding
\cite{SchmidtMontgomery1995,Burbank1996ThresholdHillslopes,Montgomery2001ThresholdHillslopes,LarsenMontgomery2012LandslideErosion}.

\section{Local geometry}

Let $\bm X=(x,z)\in\R^2$, with the cusp point placed at the origin.
In this section $z\ge0$ is measured along the local cusp axis and $x$ is
transverse to it.
The cusp half-width is
\begin{equation}
    b(z)=B z^m,\qquad B>0,\qquad m>1 .
    \label{eq:width}
\end{equation}
Equivalently, if the same cusp is represented by a graph
\begin{equation}
    z=A\abs{x}^{\alpha},
    \qquad 0<\alpha<1,
    \label{eq:graph-cusp}
\end{equation}
then
\begin{equation}
    m=\frac1\alpha,\qquad B=A^{-1/\alpha}.
    \label{eq:m-alpha}
\end{equation}
The case $m=1$ is a wedge.
The case $m>1$ is genuinely cuspidal because
\begin{equation}
    \chi(z)=\frac{b(z)}{z}=B z^{m-1}\to0,
    \qquad z\downarrow0 .
    \label{eq:slenderness}
\end{equation}

We distinguish two local domains.
The cuspidal ridge is the material horn
\begin{equation}
    \mathcal C_+
    =
    \{(x,z):0<z<\ell,\ \abs{x}<b(z)\}.
    \label{eq:ridge-domain}
\end{equation}
The cuspidal gorge is the complement of a cuspidal notch,
\begin{equation}
    \mathcal C_-
    =
    B_\ell(0)\setminus
    \{(x,z):0<z<\ell,\ \abs{x}<b(z)\}.
    \label{eq:gorge-domain}
\end{equation}
Figure~\ref{fig:local-geometries} shows the two geometries.

\begin{figure}[t]
    \centering
    \includegraphics[width=0.95\textwidth]{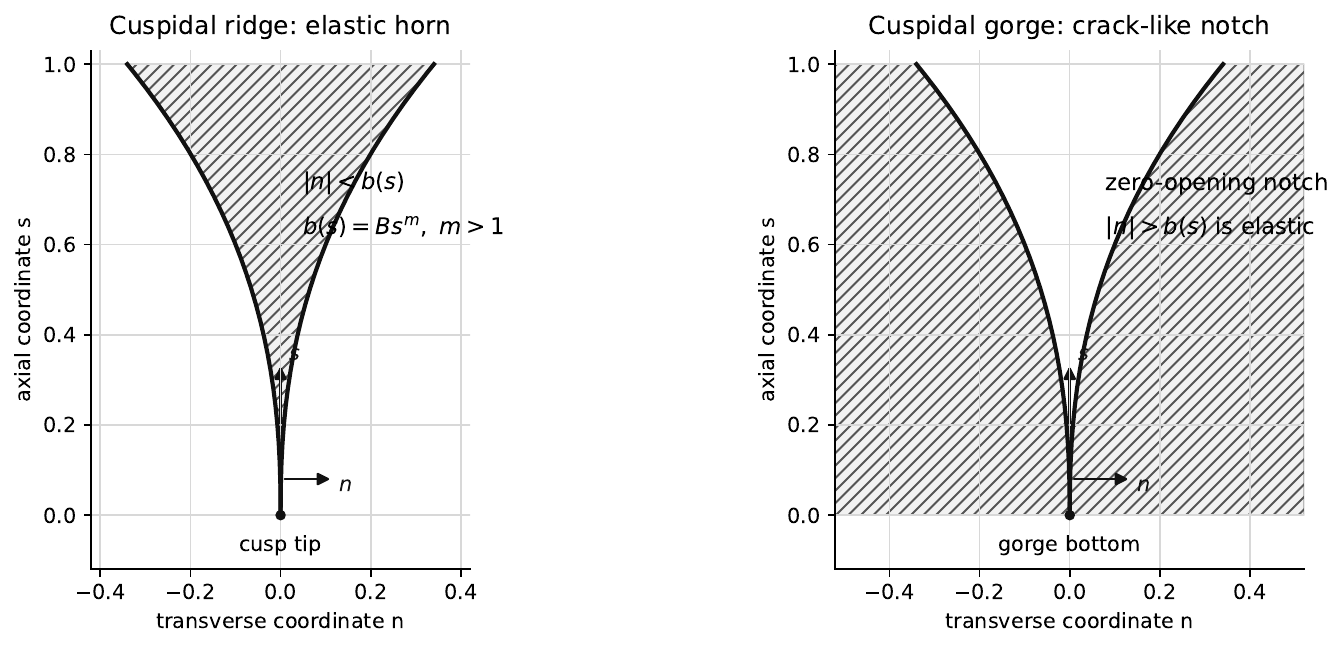}
    \caption{Local cusp geometries. Left: a cuspidal ridge is a vanishing-width elastic horn. Right: a cuspidal gorge is a zero-opening re-entrant notch.}
    \label{fig:local-geometries}
\end{figure}

\subsection{Frenet representation of the cusp branches}

The natural coordinates used below are not introduced for an arbitrary smooth curve only; they are attached to the two smooth branches of the chosen cusp.
For the graph form
\[
    z=A|x|^\alpha,\qquad 0<\alpha<1,
\]
write the two branches as
\begin{equation}
    \bm{\Gamma}_\delta(q)=(\delta q,Aq^\alpha),
    \qquad q>0,\qquad \delta=\pm1.
    \label{eq:cusp-branch-parametrization}
\end{equation}
Here $q=|x|$ is a branch parameter, not the arclength.
Set
\begin{equation}
    p(q)=A\alpha q^{\alpha-1},
    \qquad
    L(q)=\sqrt{1+p(q)^2}.
    \label{eq:cusp-p-L}
\end{equation}
The arclength from the cusp point to $\bm{\Gamma}_\delta(q)$ is
\begin{equation}
    s(q)=\int_0^q L(\eta)\,\dd\eta .
    \label{eq:cusp-arclength}
\end{equation}
As $q\downarrow0$,
\begin{equation}
    s(q)=Aq^\alpha+
    \frac{q^{2-\alpha}}{2A\alpha(2-\alpha)}
    +o(q^{2-\alpha}).
    \label{eq:cusp-arclength-asymptotic}
\end{equation}
With the orientation away from the cusp, the Frenet frame on the branch is
\begin{equation}
    \bm{\tau}_\delta(q)
    =
    \frac{(\delta,p(q))}{L(q)},
    \qquad
    \bm{\nu}_\delta(q)
    =
    \frac{(-\delta p(q),1)}{L(q)}.
    \label{eq:cusp-tangent-normal}
\end{equation}
This choice gives the Frenet formulas
\[
    \frac{\dd \bm{\tau}_\delta}{\dd s}
    =
    \kappa_\delta(s)\bm{\nu}_\delta,
    \qquad
    \frac{\dd \bm{\nu}_\delta}{\dd s}
    =
    -\kappa_\delta(s)\bm{\tau}_\delta .
\]
The signed curvature is
\begin{equation}
    \kappa_\delta(s(q))
    =
    \frac{A\alpha(\alpha-1)q^{\alpha-2}}
    {\left(1+A^2\alpha^2q^{2\alpha-2}\right)^{3/2}}.
    \label{eq:cusp-signed-curvature}
\end{equation}
It is independent of $\delta$ for the above orientation.
Near the cusp this can be written as
\begin{equation}
    \kappa_\delta(s(q))
    \sim
    \frac{\alpha-1}{A^2\alpha^2}q^{1-2\alpha},
    \qquad q\downarrow0.
    \label{eq:cusp-curvature-asymptotic}
\end{equation}
Thus the curvature is singular for $\alpha>1/2$, finite to leading order for $\alpha=1/2$, and vanishes for $\alpha<1/2$, although the tangent still has a cusp limit.

The natural coordinates around a chosen branch are therefore
\begin{equation}
    \bm{X}_\delta(s,n)
    =
    \bm{\Gamma}_\delta(q(s))
    +
    n\bm{\nu}_\delta(q(s)),
    \qquad
    H_\delta(s,n)
    =
    1-\kappa_\delta(s)n .
    \label{eq:cusp-natural-coordinates}
\end{equation}
Here $s$ is arclength along the cusp branch and $n$ is signed normal distance.
The coordinate map is locally non-degenerate as long as $H_\delta(s,n)\ne0$.
Its coordinate vectors are $(\bm{X}_\delta)_s=H_\delta\bm{\tau}_\delta$ and
$(\bm{X}_\delta)_n=\bm{\nu}_\delta$, so
\[
    |\dd\bm{X}|^2=H_\delta^2\,\dd s^2+\dd n^2.
\]

The wave problem attaches a distinguished length to this geometry.  The
surface is stationary, so the angular frequency $\omega$ is prescribed by the
incident wave and is not changed by the topography.  The flat half-space
Rayleigh wavelength at this frequency is
\begin{equation}
    \lamR=\frac{2\pi c_R}{\omega},
    \label{eq:rayleigh-wavelength}
\end{equation}
where $c_R$ is the Rayleigh phase velocity for a flat traction-free isotropic
half-space.  Topography changes the spatial field at the same frequency:
phase, amplitude, local wave numbers and mode conversion may change, while the
time factor remains fixed.  We therefore allow the Rayleigh part of the local
field to have an effective wavelength
\begin{equation}
    m_\lambda=\frac{\lambda_R}{\lambda_{\rm eff}},
    \qquad
    \frac{\lambda_{\rm eff}}{\lambda_R}=\frac1{m_\lambda}.
    \label{eq:effective-wavelength-ratio}
\end{equation}
The dependence of $m_\lambda$ on the cusp geometry is determined later from
the solvability conditions for the outer expansion.

In the next section we suppress the branch index $\delta$ and write
$H(s,n)=1-\kappa(s)n$.

\section{Rayleigh equations in natural surface coordinates}

All equations in this section are written for the two-dimensional plane-strain
model.  The spatial point is
$\bm X=(X_1,X_2)\in\mathbb R^2$, the elastic body is
$\Om\subset\mathbb R^2$, and $t$ denotes time.  The parameter
$\omega>0$ is the angular frequency.  The physical displacement is the real
vector field
\[
    \bm u:\Om\times\mathbb R\to\mathbb R^2 .
\]
For a time-harmonic Rayleigh wave we write it through a complex vector
amplitude $\bm U(\cdot;\omega):\Om\to\mathbb C^2$ as
\[
    \bm u(\bm X,t)
    =
    \operatorname{Re}\{\bm U(\bm X;\omega)e^{-i\omega t}\}
    =
    \frac12\left(
        \bm U(\bm X;\omega)e^{-i\omega t}
        +
        \overline{\bm U(\bm X;\omega)}e^{i\omega t}
    \right).
\]
Thus $\bm U$ is a two-component complex displacement amplitude, not the
physical displacement itself.  The second term in the last formula is the
complex conjugate negative-frequency component required by the reality of
$\bm u$.
In this local-coordinate derivation, bold symbols denote vectors or vector-valued maps:
$\bm X,\bm U,\bm u,\bm{\tau},\bm{\nu},\bm n$.
Plain italic symbols such as $u$, $v$, $s$, $n$, $H$, $\kappa$ and $\omega$
are scalars.  In particular, $n$ is the scalar normal coordinate, whereas
$\bm n$ is the unit normal vector.
The stress and strain recovered from $\bm U$ are written as $\sig(\bm U)$
and $e(\bm U)$; without indices they are $2\times2$ tensors, while
$\sigma_{ij}$ and $e_{ij}$ denote their scalar components.
The displacement amplitude satisfies the Navier equation
\begin{equation}
    \mu\Delta\bm U+(\lambda+\mu)\nabla(\divg\bm U)
    +\rho\omega^2\bm U=0,
    \qquad \bm X\in\Om,
    \label{eq:dynamic-lame}
\end{equation}
The velocity $c_R$ in \eqref{eq:rayleigh-wavelength} is determined by the flat
traction-free half-space.  With
\[
    c_P=\sqrt{\frac{\lambda+2\mu}{\rho}},
    \qquad
    c_S=\sqrt{\frac{\mu}{\rho}},
\]
it is the root $0<c_R<c_S$ of the Rayleigh equation
\[
    \left(2-\frac{c_R^2}{c_S^2}\right)^2
    =
    4
    \sqrt{1-\frac{c_R^2}{c_P^2}}
    \sqrt{1-\frac{c_R^2}{c_S^2}} .
\]
After all lengths are normalized by $\lambda_R$, \eqref{eq:dynamic-lame}
becomes
\begin{equation}
    \mu\Delta\bm U
    +
    (\lambda+\mu)\nabla(\divg\bm U)
    +
    4\pi^2\rho c_R^2\bm U=0.
    \label{eq:outer-wave-scale}
\end{equation}
Thus the normalized outer problem is still elastodynamic; the static equations
appear only after a further inner scaling near the cusp.
The stress is then reconstructed by the constitutive relation
\begin{equation}
    \sig(\bm U)=\lambda(\divg\bm U)I+2\mu e(\bm U),
    \qquad
    e(\bm U)=\frac12(\nabla\bm U+(\nabla\bm U)^T).
    \label{eq:stress}
\end{equation}
Here $e(\bm U)$ and $\sig(\bm U)$ are $2\times2$ strain and stress tensors,
$I$ is the $2\times2$ identity matrix, $\lambda$ and $\mu$ are the Lame
moduli, and $\rho$ is the mass density.
The free surface is traction-free:
\begin{equation}
    \sig(\bm U)\bm n=0 .
    \label{eq:traction-free}
\end{equation}

\subsection{Displacement equations}

The displacement is decomposed as
\begin{equation}
    \bm U(s,n)=u(s,n)\bm{\tau}(s)+v(s,n)\bm{\nu}(s).
    \label{eq:surface-displacement-components}
\end{equation}
In this formula $u$ and $v$ are scalar complex amplitudes of the tangential
and normal components of the vector $\bm U$.
The strain components are
\begin{equation}
\begin{aligned}
    e_{ss}&=\frac{u_s-\kappa v}{H},\\
    e_{nn}&=v_n,\\
    2e_{sn}&=u_n+\frac{v_s+\kappa u}{H}.
\end{aligned}
    \label{eq:surface-strains}
\end{equation}
We now insert the Rayleigh-wavelength normalization into these natural
coordinates.  The local graph is written as
\begin{equation}
    z=a f(x),\qquad f(x)=\abs{x}^{\alpha},
    \qquad
    a=A\lamR^{\alpha-1}.
    \label{eq:dimensionless-cusp-a}
\end{equation}
The parameter $a$ links the geometric coefficient $A$ to the Rayleigh
wavelength.  Long waves correspond to small $a$.  On each smooth branch,
\begin{equation}
    s_a(x)=\int_0^x\sqrt{1+a^2f'(\eta)^2}\,\dd\eta,
    \label{eq:arc-length-a}
\end{equation}
\begin{equation}
    \kappa_a(s_a(x))
    =
    \frac{a f''(x)}
    {\left(1+a^2f'(x)^2\right)^{3/2}},
    \qquad
    h_a(s,n)=1-\kappa_a(s)n .
    \label{eq:curvature-with-a}
\end{equation}
The effective wavelength ratio introduced in
\eqref{eq:effective-wavelength-ratio} now becomes a geometry-dependent scalar
\begin{equation}
    m_\lambda=m_a
    =
    1+\sum_{j=1}^{\infty}a^j m_j .
    \label{eq:wavenumber-ratio-series}
\end{equation}
The coefficients $m_j$ are numbers, not functions of $s$; they are determined
order by order from the Fredholm solvability conditions.  Since the coordinate
$s$ is normalized by the reference wavelength $\lambda_R$ rather than by
$\lambda_{\rm eff}$, the tangential derivative acting on the Rayleigh-scale
outer field is
\begin{equation}
    D_a:=m_a\partial_s .
    \label{eq:effective-tangential-derivative}
\end{equation}
Put
\[
    \Omega_R=4\pi^2\rho c_R^2 .
\]
Thus
\begin{equation}
\begin{aligned}
    \sig^a_{ss}
    &=
    (\lambda+2\mu)\frac{D_au-\kappa_a v}{h_a}
    +
    \lambda v_n,
    \\
    \sig^a_{nn}
    &=
    \lambda\frac{D_au-\kappa_a v}{h_a}
    +
    (\lambda+2\mu)v_n,
    \\
    \sig^a_{sn}
    &=
    \mu
    \left(
        u_n+\frac{D_av+\kappa_a u}{h_a}
    \right).
\end{aligned}
    \label{eq:stress-with-a}
\end{equation}
Substitution into the normalized Navier equation gives the Rayleigh system
for the displacement components:
\begin{equation}
\begin{aligned}
    &\frac1{h_a}D_a
    \left[
        (\lambda+2\mu)\frac{D_au-\kappa_a v}{h_a}
        +
        \lambda v_n
    \right]
    +
    \partial_n
    \left[
        \mu\left(
            u_n+\frac{D_av+\kappa_a u}{h_a}
        \right)
    \right]
    \\
    &\qquad
    -
    \frac{2\kappa_a\mu}{h_a}
    \left(
        u_n+\frac{D_av+\kappa_a u}{h_a}
    \right)
    +
    \Omega_Ru
    =
    0,
    \\
    &\frac1{h_a}D_a
    \left[
        \mu\left(
            u_n+\frac{D_av+\kappa_a u}{h_a}
        \right)
    \right]
    +
    \partial_n
    \left[
        \lambda\frac{D_au-\kappa_a v}{h_a}
        +
        (\lambda+2\mu)v_n
    \right]
    \\
    &\qquad
    +
    \frac{2\mu\kappa_a}{h_a}
    \left(
        \frac{D_au-\kappa_a v}{h_a}
        -
        v_n
    \right)
    +
    \Omega_Rv
    =
    0.
\end{aligned}
    \label{eq:rayleigh-equations-with-a}
\end{equation}
At $n=0$ the traction-free condition $\sig(\bm U)\bm n=0$ becomes, in
displacement variables,
\begin{equation}
    \mu(u_n+D_av+\kappa_a u)=0,
    \qquad
    \lambda(D_au-\kappa_a v)+(\lambda+2\mu)v_n=0
    \quad (n=0).
    \label{eq:free-boundary-with-a}
\end{equation}
\section{Outer expansion away from the cusp}

For fixed $s\ne0$,
\begin{equation}
    \kappa_a(s)=a f''(s)+O(a^3),\qquad
    h_a(s,n)=1-a f''(s)n+O(a^3).
    \label{eq:small-a-geometry}
\end{equation}
At the same time,
\begin{equation}
    D_a
    =
    \partial_s+\sum_{j=1}^{\infty}a^j m_j\partial_s
    =
    \partial_s+a m_1\partial_s+O(a^2).
    \label{eq:Da-small-a}
\end{equation}
The leading equations are therefore the flat Rayleigh half-space equations
\begin{equation}
\begin{aligned}
    (\lambda+2\mu)u_{ss}
    +
    \mu u_{nn}
    +
    (\lambda+\mu)v_{sn}
    +
    \Omega_Ru
    &=0,
    \\
    \mu v_{ss}
    +
    (\lambda+2\mu)v_{nn}
    +
    (\lambda+\mu)u_{sn}
    +
    \Omega_Rv
    &=0,
\end{aligned}
    \label{eq:plane-lame-components}
\end{equation}
with
\begin{equation}
    \mu(u_n+v_s)=0,\qquad
    \lambda u_s+(\lambda+2\mu)v_n=0
    \quad (n=0).
    \label{eq:flat-rayleigh-boundary}
\end{equation}
The passage to the flat problem is non-uniform at the cusp.
Indeed, for $f(s)=|s|^\alpha$,
\begin{equation}
    f'(s)=\alpha\,\sgn(s)|s|^{\alpha-1},
    \qquad
    f''(s)=\alpha(\alpha-1)|s|^{\alpha-2},
    \label{eq:cusp-derivatives-singular}
\end{equation}
so a regular small-$a$ expansion requires at least
\begin{equation}
    a|s|^{\alpha-1}\ll1,
    \qquad
    a|s|^{\alpha-2}|n|\ll1 .
    \label{eq:cusp-small-a-validity}
\end{equation}

It is useful to write the operator expansion explicitly.
Set
\begin{equation}
    r_a=\frac1{h_a},\qquad
    \gamma_a=\frac{\kappa_a}{h_a}.
    \label{eq:ra-gammaa-definitions}
\end{equation}
For fixed $s\ne0$,
\begin{equation}
    r_a=\sum_{j=0}^\infty a^j r_j,\qquad
    \gamma_a=\sum_{j=0}^\infty a^j \gamma_j .
    \label{eq:geometry-taylor-series}
\end{equation}
The first coefficients are
\begin{equation}
    r_0=1,\quad \gamma_0=0,\qquad
    r_1=n k(s),\quad \gamma_1=k(s),
    \label{eq:r1-gamma1}
\end{equation}
where
\begin{equation}
    k(s)=f''(s)=\alpha(\alpha-1)|s|^{\alpha-2}.
    \label{eq:k-leading-curvature}
\end{equation}
We also write
\begin{equation}
    \kappa_a(s)=\sum_{j=1}^{\infty}a^j\kappa_j(s),
    \qquad
    \kappa_1(s)=k(s).
    \label{eq:curvature-taylor-series}
\end{equation}
For a vector $W=(u,v)^T$ define the full displacement operator
$\mathcal L_a W=((\mathcal L_a W)_1,(\mathcal L_a W)_2)^T$ by the left-hand
side of \eqref{eq:rayleigh-equations-with-a}:
\begin{equation}
\begin{aligned}
    (\mathcal L_aW)_1
    &=
    r_aD_a
    \left[
        (\lambda+2\mu)r_a(D_au-\kappa_a v)
        +
        \lambda v_n
    \right]
    +
    \partial_n
    \left[
        \mu\left(u_n+r_a(D_av+\kappa_a u)\right)
    \right]
    \\
    &\quad
    -
    2\mu\gamma_a
    \left(u_n+r_a(D_av+\kappa_a u)\right)
    +
    \Omega_Ru,
    \\
    (\mathcal L_aW)_2
    &=
    r_aD_a
    \left[
        \mu\left(u_n+r_a(D_av+\kappa_a u)\right)
    \right]
    +
    \partial_n
    \left[
        \lambda r_a(D_au-\kappa_a v)
        +
        (\lambda+2\mu)v_n
    \right]
    \\
    &\quad
    +
    2\mu\gamma_a
    \left[
        r_a(D_au-\kappa_a v)-v_n
    \right]
    +
    \Omega_Rv.
\end{aligned}
    \label{eq:full-operator-La}
\end{equation}
The boundary operator at the free surface is
\begin{equation}
    \mathcal B_aW
    =
    \left.
    \begin{pmatrix}
        \mu(u_n+D_av+\kappa_a u)\\
        \lambda(D_au-\kappa_a v)+(\lambda+2\mu)v_n
    \end{pmatrix}
    \right|_{n=0}.
    \label{eq:full-boundary-operator-Ba}
\end{equation}
The coefficient operators are defined by the formal Taylor expansions
\begin{equation}
    \mathcal L_a=\sum_{p=0}^{\infty}a^p\mathcal L_p,
    \qquad
    \mathcal B_a=\sum_{p=0}^{\infty}a^p\mathcal B_p .
    \label{eq:operator-taylor-series}
\end{equation}
Equivalently, $\mathcal L_pW$ and $\mathcal B_pW$ are the coefficients of
$a^p$ in \eqref{eq:full-operator-La} and \eqref{eq:full-boundary-operator-Ba}.
In particular,
\begin{equation}
    \mathcal L_0W
    =
    \begin{pmatrix}
        (\lambda+2\mu)u_{ss}+\mu u_{nn}+(\lambda+\mu)v_{sn}+\Omega_Ru\\
        \mu v_{ss}+(\lambda+2\mu)v_{nn}+(\lambda+\mu)u_{sn}+\Omega_Rv
    \end{pmatrix},
    \label{eq:L0-explicit}
\end{equation}
and
\begin{equation}
    \mathcal B_0W
    =
    \left.
    \begin{pmatrix}
        \mu(u_n+v_s)\\
        \lambda u_s+(\lambda+2\mu)v_n
    \end{pmatrix}
    \right|_{n=0}.
    \label{eq:B0-explicit}
\end{equation}
Here $\mathcal L_0,\mathcal B_0$ already contain the leading value $m_0=1$ of
the wavelength ratio through $D_a=\partial_s+O(a)$.
The geometry-only part of the first-order coefficient is
\begin{equation}
\begin{aligned}
    (\mathcal L_1^{\rm geom}W)_1
    &=
    nk\,\partial_s\bigl((\lambda+2\mu)u_s+\lambda v_n\bigr)
    +
    \partial_s\bigl((\lambda+2\mu)k(nu_s-v)\bigr)
    \\
    &\quad
    +
    \partial_n\bigl(\mu k(nv_s+u)\bigr)
    -
    2\mu k(u_n+v_s),
    \\
    (\mathcal L_1^{\rm geom}W)_2
    &=
    nk\,\partial_s\bigl(\mu(u_n+v_s)\bigr)
    +
    \partial_s\bigl(\mu k(nv_s+u)\bigr)
    \\
    &\quad
    +
    \partial_n\bigl(\lambda k(nu_s-v)\bigr)
    +
    2\mu k(u_s-v_n),
\end{aligned}
    \label{eq:L1-geom-explicit}
\end{equation}
while
\begin{equation}
    \mathcal B_1^{\rm geom}W
    =
    \left.
    \begin{pmatrix}
        \mu k u\\
        -\lambda k v
    \end{pmatrix}
    \right|_{n=0}.
    \label{eq:B1-geom-rayleigh}
\end{equation}

The part of the first-order coefficient generated by the wavelength ratio
$m_a=1+a m_1+O(a^2)$ is isolated as follows.  If $\nu$ is a scalar first
variation of this ratio, define
\begin{equation}
\begin{aligned}
    (\mathcal P_\nu W)_1
    &=
    \nu\,\partial_s\bigl((\lambda+2\mu)u_s+\lambda v_n\bigr)
    +
    \partial_s\bigl((\lambda+2\mu)\nu u_s\bigr)
    +
    \partial_n\bigl(\mu\nu v_s\bigr),
    \\
    (\mathcal P_\nu W)_2
    &=
    \nu\,\partial_s\bigl(\mu(u_n+v_s)\bigr)
    +
    \partial_s\bigl(\mu\nu v_s\bigr)
    +
    \partial_n\bigl(\lambda\nu u_s\bigr),
\end{aligned}
    \label{eq:wavenumber-first-operator}
\end{equation}
and
\begin{equation}
    \mathcal Q_\nu W
    =
    \left.
    \begin{pmatrix}
        \mu\nu v_s\\
        \lambda\nu u_s
    \end{pmatrix}
    \right|_{n=0}.
    \label{eq:wavenumber-first-boundary}
\end{equation}
Then
\begin{equation}
    \mathcal L_1W=\mathcal L_1^{\rm geom}W+\mathcal P_{m_1}W,
    \qquad
    \mathcal B_1W=\mathcal B_1^{\rm geom}W+\mathcal Q_{m_1}W .
    \label{eq:first-order-splitting}
\end{equation}

Let
\begin{equation}
    W^a(s,n)\sim\sum_{j=0}^{\infty}a^jW_j(s,n),
    \qquad W_j=(u_j,v_j)^T .
    \label{eq:displacement-a-series}
\end{equation}
This is a complex-amplitude expansion.  The corresponding real physical
displacement components in the moving basis are
\[
    W^a_{\rm phys}(s,n,t)
    =
    \frac12
    \sum_{j=0}^{\infty}a^j
    \left(
        W_j(s,n)e^{-i\omega t}
        +
        \overline{W_j(s,n)}e^{i\omega t}
    \right).
\]
The actual vector displacement is obtained from these real components by
multiplying by the real basis vectors $\bm{\tau}(s)$ and $\bm{\nu}(s)$.
After substituting this series and \eqref{eq:wavenumber-ratio-series} into
\eqref{eq:rayleigh-equations-with-a} and \eqref{eq:free-boundary-with-a}, one obtains
\begin{equation}
    \mathcal L_0W_0=0,\qquad
    \mathcal B_0W_0=0,
    \label{eq:a-leading-problem}
\end{equation}
\begin{equation}
    \mathcal L_0W_1=-\mathcal L_1W_0,\qquad
    \mathcal B_0W_1=-\mathcal B_1W_0,
    \label{eq:a-first-correction-problem}
\end{equation}
Equivalently, the first correction solves
\begin{equation}
    \mathcal L_0W_1=F_1,\qquad
    \mathcal B_0W_1=G_1,
    \label{eq:first-correction-inhomogeneous}
\end{equation}
where
\begin{equation}
    F_1=F_1^{\rm geom}-\mathcal P_{m_1}W_0,\qquad
    G_1=G_1^{\rm geom}-\mathcal Q_{m_1}W_0,
    \label{eq:first-correction-data}
\end{equation}
and
\[
    F_1^{\rm geom}=-\mathcal L_1^{\rm geom}W_0,\qquad
    G_1^{\rm geom}=-\mathcal B_1^{\rm geom}W_0 .
\]
The geometry-only interior data are
\begin{equation}
\begin{aligned}
    (F_1^{\rm geom})_1
    &=
    -nk\,\partial_s\bigl((\lambda+2\mu)u_{0,s}+\lambda v_{0,n}\bigr)
    -
    \partial_s\bigl((\lambda+2\mu)k(nu_{0,s}-v_0)\bigr)
    \\
    &\quad
    -
    \partial_n\bigl(\mu k(nv_{0,s}+u_0)\bigr)
    +
    2\mu k(u_{0,n}+v_{0,s}),
    \\
    (F_1^{\rm geom})_2
    &=
    -nk\,\partial_s\bigl(\mu(u_{0,n}+v_{0,s})\bigr)
    -
    \partial_s\bigl(\mu k(nv_{0,s}+u_0)\bigr)
    \\
    &\quad
    -
    \partial_n\bigl(\lambda k(nu_{0,s}-v_0)\bigr)
    -
    2\mu k(u_{0,s}-v_{0,n}),
\end{aligned}
    \label{eq:first-correction-F1-geom}
\end{equation}
and the geometry-only boundary data are
\begin{equation}
    G_1^{\rm geom}
    =
    \left.
    \begin{pmatrix}
        -\mu k u_0\\
        \lambda k v_0
    \end{pmatrix}
    \right|_{n=0}.
    \label{eq:first-correction-G1-geom}
\end{equation}

\begin{lemma}[Fredholm alternative for the first correction]
\label{lem:first-correction-solvability}
Let $0<\delta<R$ and let $D_{\delta,R,H}$ be a finite strip of the flat
half-space in natural coordinates,
\[
    D_{\delta,R,H}=\{(s,n):\delta<|s|<R,\ 0<n<H\},
\]
with the free boundary at $n=0$ and with fixed homogeneous auxiliary
conditions on the artificial sides.
If $W_0$ is a smooth solution of the flat Rayleigh problem
\eqref{eq:a-leading-problem} in a neighbourhood of $\overline D_{\delta,R,H}$,
then the first-correction problem \eqref{eq:first-correction-inhomogeneous}
is solvable if and only if its data satisfy the Fredholm compatibility
conditions
\begin{equation}
    \int_{D_{\delta,R,H}} F_1\cdot\overline{\Phi}\,\dd s\,\dd n
    +
    \int_{\Gamma_0}G_1\cdot\overline{\Phi}\,\dd s
    =
    0
    \qquad
    \text{for all }\Phi\in\mathcal N^* .
    \label{eq:first-correction-fredholm-condition}
\end{equation}
Here $\Gamma_0=\{(s,0):\delta<|s|<R\}$ and $\mathcal N^*$ is the finite
dimensional kernel of the adjoint homogeneous flat Lam\'e problem with the
corresponding homogeneous adjoint auxiliary conditions.
If the compatibility conditions hold, the solution is determined up to the
kernel of the direct homogeneous flat problem.  This remaining homogeneous
part is fixed by the chosen radiation and normalization conditions.  In
particular, if the adjoint kernel is trivial after those conditions are
imposed, then the compatibility conditions are empty.
\end{lemma}

\begin{proof}
For fixed $\delta>0$, the functions $k(s)$, $r_j(s,n)$ and $\gamma_j(s,n)$
are smooth and bounded on $\overline D_{\delta,R,H}$.
Hence the data $F_1$ and $G_1$ in
\eqref{eq:first-correction-data}, with the explicit components
\eqref{eq:first-correction-F1-geom}--\eqref{eq:first-correction-G1-geom} and
\eqref{eq:wavenumber-first-operator}--\eqref{eq:wavenumber-first-boundary}, have the same
local Sobolev regularity as the corresponding derivatives of $W_0$.
The pair $(\mathcal L_0,\mathcal B_0)$ is the constant-coefficient Lam\'e
system with the traction-free boundary operator; it is elliptic and satisfies
the complementing boundary condition on the flat boundary.
With the auxiliary side conditions it defines a Fredholm map between the
corresponding Sobolev spaces.
Green's identity for the Lam\'e operator gives, for any adjoint homogeneous
mode $\Phi\in\mathcal N^*$,
\[
    \int_{D_{\delta,R,H}}(\mathcal L_0 W_1)\cdot\overline{\Phi}\,\dd s\,\dd n
    +
    \int_{\Gamma_0}(\mathcal B_0 W_1)\cdot\overline{\Phi}\,\dd s
    =
    0,
\]
where the artificial-side terms vanish because of the paired homogeneous
boundary conditions.  Substitution of
\eqref{eq:first-correction-inhomogeneous} gives
\eqref{eq:first-correction-fredholm-condition}; hence the condition is
necessary.
Conversely, the Fredholm alternative for the elliptic boundary-value problem
states that this orthogonality to the adjoint kernel is sufficient for
solvability.  If the adjoint kernel is trivial, there is no compatibility
condition to check.  Non-uniqueness is exactly the addition of homogeneous
solutions of the flat problem, and is removed by fixing the homogeneous
Rayleigh-mode amplitudes or by imposing the chosen radiation/normalization
condition.
\end{proof}

Thus the special right-hand side in
\eqref{eq:first-correction-data} is not automatically admissible in a resonant
truncated problem.  The orthogonality condition
\eqref{eq:first-correction-fredholm-condition} is the equation that determines
the scalar first correction $m_1$ to the effective Rayleigh wavelength
parameter.

For the positive-frequency phasor of a single right-going flat Rayleigh mode
this conclusion can be checked directly.  The projection constants are
computed for the unit phasor.  Let $q=2\pi$ and write
\[
    W_0(s,n)=e^{iqs}\Phi(n),
\]
where
\[
    \Phi(n)=
    \begin{pmatrix}
        i\left(e^{-pqn}-\dfrac{2p\beta}{1+\beta^2}e^{-\beta qn}\right)\\[4pt]
        -p e^{-pqn}+\dfrac{2p}{1+\beta^2}e^{-\beta qn}
    \end{pmatrix},
    \qquad
    p=\sqrt{1-\frac{c_R^2}{c_P^2}},\quad
    \beta=\sqrt{1-\frac{c_R^2}{c_S^2}},
\]
and the Rayleigh equation is $(1+\beta^2)^2=4p\beta$.
The physical leading displacement represented by this phasor is not $W_0$
alone.  With a complex amplitude $A$ it is
\begin{equation}
    W_{0,\rm phys}(s,n,t)
    =
    \frac12\left(
        A e^{i(qs-\omega t)}\Phi(n)
        +
        \overline{A} e^{-i(qs-\omega t)}\overline{\Phi(n)}
    \right).
    \label{eq:real-rayleigh-mode}
\end{equation}
The conjugate component satisfies the conjugate correction problem, and the
constant amplitude cancels from the scalar equation for $m_1$.
Testing the geometry-only first-correction data against the adjoint mode
$e^{iqs}\Phi(n)$ gives the Fredholm density
\[
    C_\kappa\,k(s)+C_{\kappa'}\,k'(s),
\]
with
\[
\begin{aligned}
    C_\kappa
    &=
    \int_0^\infty F_\kappa(n)\cdot\overline{\Phi(n)}\,\dd n
    +
    G_\kappa\cdot\overline{\Phi(0)},\\
    C_{\kappa'}
    &=
    \int_0^\infty F_{\kappa'}(n)\cdot\overline{\Phi(n)}\,\dd n .
\end{aligned}
\]
Here $F_\kappa,G_\kappa$ are obtained by setting $k=1$, $k'=0$, and
$F_{\kappa'}$ by setting $k=0$, $k'=1$ in the data above.  The integrations are
only in $n$; the remaining integration over $s$ averages $k(s)$ and $k'(s)$.
For $q=1$, $\mu=1$, $\lambda=2$, direct integration gives
\[
    C_\kappa=5.920160973\ldots,\qquad
    C_{\kappa'}=-3.640899899\ldots\, i .
\]
The scalar wavelength correction is obtained from the real part of the phasor
Fredholm condition; the imaginary part belongs to the quadrature
phase/amplitude balance of the complex Rayleigh amplitude.
The wavelength degree of freedom contributes another projection.  Let
$D_m$ denote the projection of
\eqref{eq:wavenumber-first-operator}--\eqref{eq:wavenumber-first-boundary}
corresponding to the scalar variation $\nu=1$:
\[
    D_m
    =
    \int_0^\infty P_m(n)\cdot\overline{\Phi(n)}\,\dd n
    +
    Q_m\cdot\overline{\Phi(0)} .
\]
Here $P_m,Q_m$ are obtained from
\eqref{eq:wavenumber-first-operator}--\eqref{eq:wavenumber-first-boundary} by
setting $\nu=1$.  The Fredholm condition is not imposed pointwise in $s$.
On a truncated outer interval
\[
    I_{\delta,R}=\{s:\delta<|s|<R\}
\]
it gives the algebraic equation
\begin{equation}
    D_m m_1 |I_{\delta,R}|
    =
    \operatorname{Re}
    \int_{I_{\delta,R}}
    \left(C_\kappa k(s)+C_{\kappa'}k'(s)\right)\,\dd s .
    \label{eq:first-wavenumber-correction}
\end{equation}
For the same normalization $q=\mu=1$, $\lambda=2$,
\[
    D_m=-2.976857206\ldots ,
\]
so the scalar solvability equation is non-degenerate.  For the constants above
$\operatorname{Re}C_{\kappa'}=0$, so the $k'$ term does not contribute to the
real wavelength correction.  Since
\[
    \lambda_{\rm eff}(a)=\lambda_R\left(1-a m_1+O(a^2)\right),
\]
the number $-m_1$ is precisely the first relative wavelength correction for the
chosen outer solvability problem.

More generally, the full recursion obtained from the $m_a$-dependent operator
expansion \eqref{eq:operator-taylor-series} is
\begin{equation}
    \mathcal L_0W_j
    =
    -\sum_{p=1}^j \mathcal L_p W_{j-p},
    \qquad
    \mathcal B_0W_j
    =
    -\sum_{p=1}^j \mathcal B_p W_{j-p},
    \qquad j\ge1.
    \label{eq:full-outer-recursion}
\end{equation}
In these sums the coefficients $\mathcal L_p,\mathcal B_p$ already contain the
contributions coming from the geometry and from the wavelength ratio
$m_a=1+\sum_{q\ge1}a^q m_q$.
Since the new scalar $m_j$ first appears only in the terms
$\mathcal L_jW_0,\mathcal B_jW_0$, one may rewrite the $j$th step as
\begin{equation}
    \mathcal L_0W_j
    =
    H_j-\mathcal P_{m_j}W_0,
    \qquad
    \mathcal B_0W_j
    =
    K_j-\mathcal Q_{m_j}W_0 .
    \label{eq:a-nth-correction-problem}
\end{equation}
Here $H_j$ and $K_j$ are explicit expressions obtained from
\eqref{eq:full-outer-recursion}; they depend only on the geometric
coefficients, on $m_1,\ldots,m_{j-1}$, and on
$W_0,\ldots,W_{j-1}$.  Projection onto the adjoint Rayleigh mode gives
\begin{equation}
    D_m m_j |I_{\delta,R}|
    =
    \operatorname{Re} R_j(I_{\delta,R}),
    \label{eq:nth-wavenumber-correction}
\end{equation}
where $R_j(I_{\delta,R})$ is the known Fredholm projection of
$(H_j,K_j)$ over the same interval.  Once $m_j$ is fixed by
\eqref{eq:nth-wavenumber-correction}, the real wavelength-channel part of the
Fredholm compatibility condition for $W_j$ is satisfied.

The scalar wavelength correction changes only the regular phase of the
right-going Rayleigh factor, for instance through the first-order term
$iqm_1sW_0$.  It does not create an additional local cusp singularity.  After
the compatibility condition has been enforced, the local singular part of the
first displacement correction is therefore the geometric rotation term
\begin{equation}
    W_1^{\rm sing}
    =
    f'(s)
    \begin{pmatrix}
        v_0-n u_{0,s}\\
        -u_0-n v_{0,s}
    \end{pmatrix}.
    \label{eq:W1-singular-part}
\end{equation}
For $f(s)=|s|^\alpha$ one has $f'(s)=O(|s|^{\alpha-1})$ and
$f''(s)=O(|s|^{\alpha-2})$.  Consequently
\begin{equation}
    W_1=O(|s|^{\alpha-1}),\qquad
    \nabla W_1=O(|s|^{\alpha-2}).
    \label{eq:W1-growth}
\end{equation}
\begin{theorem}[formal solvability of the modulated outer expansion]
\label{thm:modulated-outer-expansion}
Assume that the adjoint Rayleigh eigenspace is one-dimensional and that
$D_m\ne0$.  On every truncated interval $\delta<|s|<R$ the formal hierarchy can
be solved recursively after the homogeneous Rayleigh amplitudes in the
corrections are fixed.  At each order $j$, the Fredholm compatibility condition
determines the scalar wavelength correction through
\eqref{eq:nth-wavenumber-correction}; with this choice of $m_j$, the
real part of the resonant obstruction to the inhomogeneous problem for $W_j$
is removed.  The quadrature part is handled by the complex amplitude/radiation
normalization and by adding the conjugate phasor in the real field.
\end{theorem}

\begin{proof}
The assertion is the Fredholm alternative applied at each order.  At order
$j$, all terms in $H_j,K_j$ are already known from lower orders, while the
unknown $m_j$ enters linearly through
$-\mathcal P_{m_j}W_0,-\mathcal Q_{m_j}W_0$.  Projection onto the adjoint
Rayleigh mode gives \eqref{eq:nth-wavenumber-correction}.  Since $D_m\ne0$,
this algebraic equation determines the real scalar $m_j$ and removes the
obstruction in the wavelength channel.  The complementary quadrature condition
fixes the usual complex homogeneous Rayleigh amplitude or, in a scattering
formulation, the corresponding radiation coefficient.
\end{proof}

\begin{proposition}[outer validity scale]
\label{prop:outer-validity-scale}
For the modulated Rayleigh expansion in the wavelength-scale boundary layer
$n=O(1)$, the regular expansion in powers of $a$ is ordered in the region
\begin{equation}
    a|s|^{\alpha-2}\ll1 .
    \label{eq:outer-validity-condition}
\end{equation}
The first transition from the flat Rayleigh outer field occurs at the inner
curvature/wavelength scale
\begin{equation}
    |s|=O(\ell_a),\qquad
    \ell_a=a^{1/(2-\alpha)} .
    \label{eq:cusp-transition-scale-a}
\end{equation}
\end{proposition}

\begin{proof}
For $f(s)=|s|^\alpha$, the curvature coefficient entering both the geometric
operator and the wavelength modulation is
$k(s)=O(|s|^{\alpha-2})$.  Hence the first small parameter of the regular
Rayleigh-scale hierarchy is $a|s|^{\alpha-2}$, not $a$ alone.  Equating this
combination to one gives \eqref{eq:cusp-transition-scale-a}.
\end{proof}

\section{Inner scaling and anisotropic blow-up of the cusp}

The scale \eqref{eq:cusp-transition-scale-a} is a longitudinal scale: it marks
where the Rayleigh-scale expansion first loses uniformity along the surface.
It does not mean that the transverse coordinate in a vanishing horn has the
same scale.  This point is essential for the ridge geometry.
We therefore set
\begin{equation}
    s=\ell_a r,\qquad
    \ell_a=a^{1/(2-\alpha)} ,
    \label{eq:inner-natural-variables}
\end{equation}
and keep the normal scale tied to the local cross-section.
In dimensional variables the longitudinal inner length is
\begin{equation}
    L_a=\lambda_R\ell_a .
    \label{eq:dimensional-inner-length-scale}
\end{equation}
For a ridge, the physical half-width at $s=\ell_a r$ is
$B(\ell_a r)^m$.  Hence the correct transverse stretching is
\begin{equation}
    n=\ell_a^m y,
    \qquad
    \mathcal C_{\rm in}^{+}
    =
    \{r>0,\ |y|<Br^m\}.
    \label{eq:inner-cusp-domain-xy}
\end{equation}
Equivalently, the section can be blown up directly by
\begin{equation}
    s=\ell_a r,\qquad
    n=B\ell_a^m r^m\eta,
    \qquad -1<\eta<1 .
    \label{eq:cusp-blowup-map}
\end{equation}
Thus the cusp point becomes the edge $r=0$, $-1\le\eta\le1$.

The displacement scale is still inherited from the outer expansion.  Let
$U_{\rm R}$ denote the characteristic Rayleigh displacement amplitude.  After
the rigid translation and rotation at the cusp have been removed,
\[
    W_0(\ell_a r,B\ell_a^m r^m\eta)
    =
    W_0(0,0)+O(U_{\rm R}\ell_a)
    \qquad (r=O(1),\ |\eta|\le1).
\]
The singular geometric correction satisfies
$W_1^{\rm sing}=O(|s|^{\alpha-1})$, and therefore
\[
    aW_1^{\rm sing}(\ell_a r,B\ell_a^m r^m\eta)
    =
    O\left(U_{\rm R}a\ell_a^{\alpha-1}\right)
    =
    O(U_{\rm R}\ell_a).
\]
The non-rigid inner displacement scale is consequently
\begin{equation}
    U_a=U_{\rm R}\ell_a .
    \label{eq:inner-displacement-scale}
\end{equation}
We write
\begin{equation}
    W^a(s,n)=W_{\rm rig}^a(s,n)+U_a\mathcal W^a(r,\eta),
    \qquad
    \mathcal W^a=(U^a,V^a)^T .
    \label{eq:inner-displacement-normalization}
\end{equation}
Here $W_{\rm rig}^a$ carries no strain.  A generic cross-sectionally varying
inner field would have transverse strain scale
$U_a/(\lambda_R\ell_a^m)$, but the transverse Neumann problem below suppresses
that variation at leading order.  The surviving zero-mode strain scale is
\begin{equation}
    e_{\rm in}=O\left(\frac{U_a}{\lambda_R\ell_a}\right),
    \qquad
    \sigma_{\rm in}=O\left(\mu\frac{U_a}{\lambda_R\ell_a}\right),
    \label{eq:inner-strain-stress-scale}
\end{equation}
that is, $e_{\rm in}=O(U_{\rm R}/\lambda_R)$ and
$\sigma_{\rm in}=O(\mu U_{\rm R}/\lambda_R)$.

For the direct blow-up \eqref{eq:cusp-blowup-map}, differentiation at fixed
physical variables gives
\begin{equation}
    \partial_s
    =
    \ell_a^{-1}
    \left(
        \partial_r-\frac{m\eta}{r}\partial_\eta
    \right),
    \qquad
    \partial_n
    =
    \frac{1}{B\ell_a^m}r^{-m}\partial_\eta .
    \label{eq:partial-x-y}
\end{equation}
For a scalar function $w$,
\begin{equation}
\begin{aligned}
    \Delta w
    &=
    \ell_a^{-2}
    \left[
        w_{rr}
        -
        \frac{2m\eta}{r}w_{r\eta}
        +
        \frac{m^2\eta^2}{r^2}w_{\eta\eta}
        +
        \frac{m(m+1)\eta}{r^2}w_\eta
    \right]
    +
    \frac{\ell_a^{-2m}}{B^2}r^{-2m}w_{\eta\eta}.
\end{aligned}
    \label{eq:blownup-laplacian}
\end{equation}
Unlike polar blow-up of a cone, the cusp operator has no single homogeneous
factor.  Since $m>1$, the transverse term
$\ell_a^{-2m}r^{-2m}\partial_{\eta\eta}$ is stronger than the longitudinal
terms on the inner longitudinal scale.

The strain components in $(r,\eta)$ are
\begin{equation}
\begin{aligned}
    e_{ss}
    &=
    \ell_a^{-1}
    \left(
        U_r-\frac{m\eta}{r}U_\eta
    \right),
    \\
    e_{nn}
    &=
    \frac{\ell_a^{-m}}{B}r^{-m}V_\eta,
    \\
    2e_{sn}
    &=
    \frac{\ell_a^{-m}}{B}r^{-m}U_\eta
    +
    \ell_a^{-1}
    \left(
        V_r-\frac{m\eta}{r}V_\eta
    \right).
\end{aligned}
    \label{eq:blownup-strain-explicit}
\end{equation}
After multiplying the Lam\'e equations by $\ell_a^{2m}$, the leading interior
terms are
\begin{equation}
    \frac{\mu}{B^2}r^{-2m}U_{\eta\eta}+\cdots=0,
    \qquad
    \frac{\lambda+2\mu}{B^2}r^{-2m}V_{\eta\eta}+\cdots=0.
    \label{eq:leading-transverse}
\end{equation}
The omitted longitudinal, mixed and inertial terms are lower order by powers
of $\ell_a^{m-1}$, $\ell_a^{2m-2}$ and $\ell_a^{2m}$.

On the sides $\eta=\delta$, $\delta=\pm1$, the traction-free conditions are
\begin{equation}
\begin{aligned}
    &-Bm(\ell_a r)^{m-1}\sigma_{ss}
    +\delta\sigma_{sn}
    =0,
    \\
    &-Bm(\ell_a r)^{m-1}\sigma_{sn}
    +\delta\sigma_{nn}
    =0 .
\end{aligned}
    \label{eq:inner-xy-free-boundary}
\end{equation}
Their leading parts are precisely Neumann conditions in the blown-up
cross-section.  Therefore the leading transverse problem has only
$\eta$-independent finite modes:
\begin{equation}
    U_0=U_0(r),\qquad V_0=V_0(r).
    \label{eq:leading-eta-independent}
\end{equation}

Equivalently, the transverse eigenproblem is
\begin{equation}
    -\phi''(\eta)=\nu\phi(\eta),
    \qquad
    \phi'(-1)=\phi'(1)=0 .
    \label{eq:eta-neumann-pencil}
\end{equation}
The zero mode is $\nu_0=0$, while for $k\ge1$,
\begin{equation}
    \nu_k=\frac{k^2\pi^2}{4},
    \qquad
    \phi_k(\eta)=
    \cos\frac{k\pi(\eta+1)}2 .
    \label{eq:eta-nonzero-modes}
\end{equation}
The non-zero modes acquire the large transverse penalty $r^{-2m}\nu_k$.
The zero-mode amplitudes satisfy, to leading order,
\begin{equation}
    \frac{\dd}{\dd r}\left(r^m U_0'\right)
    +
    \frac{\eps_a^2\rho c_R^2}{\lambda+2\mu}r^mU_0
    =
    0,
    \label{eq:U0-radial-ode}
\end{equation}
\begin{equation}
    \frac{\dd}{\dd r}\left(r^m V_0'\right)
    +
    \frac{\eps_a^2\rho c_R^2}{\mu}r^mV_0
    =
    0.
    \label{eq:V0-radial-ode}
\end{equation}
For $\eps_a=0$ the independent branches are constants and $r^{1-m}$.

\section{Cuspidal ridge as a finite-strength elastic horn limit}

The static inner problem in the ridge is
\begin{equation}
    \mu\Delta\bm U+(\lambda+\mu)\nabla(\divg\bm U)=0
    \qquad \text{in }\mathcal C_+,
    \label{eq:static-ridge-full}
\end{equation}
with traction-free lateral sides, where the traction is computed from
$\sig(\bm U)$.
For the section $I_s=\{(s,n):|n|<b(s)\}$ define the force resultants
\begin{equation}
    N(s)=\int_{-b(s)}^{b(s)}\sigma_{ss}(s,n)\,\dd n,
    \qquad
    Q(s)=\int_{-b(s)}^{b(s)}\sigma_{sn}(s,n)\,\dd n,
    \label{eq:force-resultants}
\end{equation}
and the bending moment
\begin{equation}
    M(s)=\int_{-b(s)}^{b(s)} n\sigma_{ss}(s,n)\,\dd n .
    \label{eq:bending-moment}
\end{equation}
The area and second moment of the section are
\begin{equation}
    A(s)=2B s^m,\qquad
    I(s)=\frac23 B^3s^{3m}.
    \label{eq:area-moment}
\end{equation}

\begin{lemma}[sectional resultants]
\label{lem:ridge-resultants}
For a traction-free static ridge,
\begin{equation}
    N(s)=N_0,\qquad Q(s)=Q_0,
    \label{eq:NQ-constant}
\end{equation}
and
\begin{equation}
    sQ(s)-M(s)=H_0 .
    \label{eq:moment-balance}
\end{equation}
\end{lemma}

\begin{proof}
Apply the divergence theorem to the part of the horn between two sections.
The lateral boundary contributes no force because it is traction-free.
The same argument applied to the moment of the traction about the cusp point gives the conserved flux $sQ-M$.
\end{proof}

\begin{proposition}[finite-strength cut-off for tip resultants]
\label{prop:finite-strength-cutoff}
Let the ideal horn be replaced by a truncated horn
\[
    \mathcal C_{+,\delta}=\{s>\delta,\ |n|<Bs^m\},
    \qquad 0<\delta\ll1,
\]
and suppose that the admissible stress level at the truncated tip is bounded
by a material strength $\sigma_\ast$.
Then the sectional resultants that can be transmitted through the terminal
section $s=\delta$ without exceeding this stress level must satisfy
\begin{equation}
    |N_0|\le 2B\sigma_\ast\delta^m,
    \qquad
    |Q_0|\le 2B\sigma_\ast\delta^m,
    \label{eq:force-strength-bound}
\end{equation}
and
\begin{equation}
    |\delta Q_0-H_0|
    \le
    \frac23 B^2\sigma_\ast\delta^{2m}.
    \label{eq:moment-strength-bound}
\end{equation}
Consequently, a horn carrying a prescribed non-zero tip force or tip moment
has a positive minimal truncation length.  In particular,
\begin{equation}
    \delta
    \ge
    \left(\frac{|N_0|}{2B\sigma_\ast}\right)^{1/m},
    \qquad
    \delta
    \ge
    \left(\frac{|Q_0|}{2B\sigma_\ast}\right)^{1/m},
    \label{eq:force-min-cutoff}
\end{equation}
and, for a pure bending resultant $Q_0=0$,
\begin{equation}
    \delta
    \ge
    \left(\frac{3|H_0|}{2B^2\sigma_\ast}\right)^{1/(2m)} .
    \label{eq:moment-min-cutoff}
\end{equation}
\end{proposition}

\begin{proof}
At the terminal section $s=\delta$, the section area and moment of inertia are
$A(\delta)=2B\delta^m$ and $I(\delta)=2B^3\delta^{3m}/3$.
An axial or shear resultant of magnitude $N_0$ or $Q_0$ requires mean stresses
of order $N_0/A(\delta)$ and $Q_0/A(\delta)$; the bound by $\sigma_\ast$ gives
\eqref{eq:force-strength-bound}.  The bending moment at this section is
$M(\delta)=\delta Q_0-H_0$.  The extreme bending stress is
\[
    \frac{|M(\delta)|B\delta^m}{I(\delta)}
    =
    \frac{3|\delta Q_0-H_0|}{2B^2\delta^{2m}},
\]
and bounding it by $\sigma_\ast$ gives \eqref{eq:moment-strength-bound}.
Solving these inequalities for $\delta$ gives the stated lower bounds.
\end{proof}

Thus the ideal cuspidal ridge is not a model of a horn that can carry an
arbitrary finite force through a point.  It is the zero-tip-load limit of
truncated horns whose stresses remain below a prescribed non-destructive
level.  If $N_0,Q_0,H_0$ are kept non-zero while $\delta\to0$, then
\eqref{eq:force-strength-bound}--\eqref{eq:moment-strength-bound} fail and the
tip must fail, yield, fracture, or be replaced by a finite truncation.
Equivalently, the limiting ideal cusp with bounded stresses has
\[
    N_0\to0,\qquad Q_0\to0,\qquad H_0\to0 .
\]

The same conclusion is consistent with the finite-energy formulation of the
ideal cusp.  Cauchy's inequality gives
\begin{equation}
    N(s)^2\le A(s)\int_{I_s}\sigma_{ss}^2\,\dd n,
    \qquad
    Q(s)^2\le A(s)\int_{I_s}\sigma_{sn}^2\,\dd n,
    \label{eq:NQ-cauchy}
\end{equation}
and
\begin{equation}
    M(s)^2\le I(s)\int_{I_s}\sigma_{ss}^2\,\dd n .
    \label{eq:M-cauchy}
\end{equation}
The stress energy is bounded below by a constant multiple of
\begin{equation}
    \int_0^\ell
    \left[
        \frac{N_0^2}{A(s)}
        +
        \frac{Q_0^2}{A(s)}
        +
        \frac{(sQ_0-H_0)^2}{I(s)}
    \right]\dd s .
    \label{eq:resultant-energy-lower-bound}
\end{equation}
Since $A(s)\asymp s^m$ and $I(s)\asymp s^{3m}$ with $m>1$, a finite-energy
ideal cusp can only be the zero-resultant limit.

If non-zero resultants are formally continued toward the ideal tip, the corresponding Saint-Venant stresses have the orders
\begin{equation}
    \sigma_{ss}^{(N)}\sim \frac{N_0}{2B}s^{-m},
    \qquad
    \sigma_{sn}^{(Q)}\sim \frac{Q_0}{2B}s^{-m},
    \label{eq:force-singular-stresses}
\end{equation}
and
\begin{equation}
    \max_{|n|\le b(s)}|\sigma_{ss}^{(M)}|
    \asymp C|H_0|s^{-2m}.
    \label{eq:bending-singular-stress}
\end{equation}
These modes are therefore better interpreted as tip-load or finite-truncation
channels, not as admissible branches of a non-destructive ideal free cusp.

\begin{proposition}[high-energy tip-load branch]
\label{prop:high-energy-tip-load-branch}
Keep the force and moment resultants $N_0,Q_0,H_0$ non-zero and truncate the
horn at $s=\delta$.  Up to material constants, the force and shear channels
have the scales
\begin{equation}
    \sigma^{(N,Q)}=O(s^{-m}),
    \qquad
    e^{(N,Q)}=O(s^{-m}),
    \qquad
    U^{(N,Q)}=O(s^{1-m}),
    \label{eq:force-branch-scales}
\end{equation}
and their elastic energy in $\{\delta<s<\ell\}$ satisfies
\begin{equation}
    E_{N,Q}(\delta)
    \asymp
    (N_0^2+Q_0^2)
    \int_\delta^\ell s^{-m}\,\dd s
    \asymp
    C_{N,Q}\delta^{1-m}.
    \label{eq:force-branch-energy}
\end{equation}
For a pure bending channel, $Q_0=0$, the corresponding scales are
\begin{equation}
\begin{aligned}
    \sigma^{(H)}&=O(s^{-2m}),
    &
    e^{(H)}&=O(s^{-2m}),
    \\
    \theta^{(H)}&=O(s^{1-3m}),
    &
    U^{(H)}&=O(s^{2-3m}),
\end{aligned}
    \label{eq:moment-branch-scales}
\end{equation}
and
\begin{equation}
    E_H(\delta)
    \asymp
    H_0^2
    \int_\delta^\ell s^{-3m}\,\dd s
    \asymp
    C_H\delta^{1-3m}.
    \label{eq:moment-branch-energy}
\end{equation}
Thus these branches are high-energy as $\delta\to0$.
\end{proposition}

\begin{proof}
The force and shear estimates follow from
$\sigma^{(N,Q)}\asymp N_0/A(s),Q_0/A(s)$ and $A(s)=2Bs^m$.
Integrating the corresponding strain gives the displacement scale
$s^{1-m}$, and integrating the section energy density gives
\eqref{eq:force-branch-energy}.  For bending,
$I(s)=2B^3s^{3m}/3$ and the curvature of the section rotation satisfies
$\theta'(s)\asymp H_0/I(s)=O(s^{-3m})$.  Hence
$\theta=O(s^{1-3m})$ and the transverse displacement is
$O(s^{2-3m})$.  The extreme bending stress is $O(s^{-2m})$, and the section
energy density is proportional to $H_0^2/I(s)=O(s^{-3m})$, which gives
\eqref{eq:moment-branch-energy}.
\end{proof}

\begin{remark}[weighted regularization]
The high-energy branches can be kept in a weighted space rather than in the
ordinary elastic energy space.  For $\beta\in\mathbb R$ define
\begin{equation}
    \mathcal H_\beta(\mathcal C_+)
    =
    \left\{
        \bm U:
        \int_{\mathcal C_+}
        s^\beta |e(\bm U)|^2\,\dd A<\infty
    \right\}.
    \label{eq:weighted-energy-space}
\end{equation}
The force and shear tip-load branches belong to $\mathcal H_\beta$ precisely
when
\begin{equation}
    \beta>m-1,
    \label{eq:force-weight-condition}
\end{equation}
while the pure moment branch belongs to $\mathcal H_\beta$ when
\begin{equation}
    \beta>3m-1 .
    \label{eq:moment-weight-condition}
\end{equation}
This is the natural functional setting for treating the tip-load branch as a
singular generalized solution.  It is not the same physical limit as the
non-destructive free-tip ridge; it is a renormalized, weighted formulation of
a truncated-tip or concentrated-load problem.  Weighted Korn inequalities and
asymptotic decompositions for elastic cusp domains provide the rigorous
background for such spaces
\cite{Nazarov2012WeightedKorn,KozlovNazarov2016ElasticSolidCusps,KozlovNazarov2018ElasticCusp}.
\end{remark}

The same conclusion is visible in the scalar horn model
\begin{equation}
    \frac{\dd}{\dd s}
    \left(
        b(s)\frac{\dd q}{\dd s}
    \right)
    +
    k^2b(s)q=0,
    \qquad b(s)=Bs^m ,
    \label{eq:horn-model}
\end{equation}
or
\begin{equation}
    q''+\frac{m}{s}q'+k^2q=0.
    \label{eq:horn-ode}
\end{equation}
Its two local branches are
\begin{equation}
    q_{\rm reg}(s)=q_0\left(1-\frac{k^2s^2}{2(m+1)}+O(s^4)\right),
    \qquad
    q_{\rm sing}(s)\sim Cs^{1-m}.
    \label{eq:q-branches}
\end{equation}
The singular branch has infinite energy:
\begin{equation}
    \int_0^\ell b(s)|q_{\rm sing}'(s)|^2\,\dd s=+\infty .
    \label{eq:horn-energy}
\end{equation}

\begin{proposition}[leading free-tip ridge field]
\label{prop:leading-ridge-field}
Within the local quasi-static, non-destructive free-tip class, the leading field in a free cuspidal ridge is asymptotically rigid:
\begin{equation}
    \bm U(s,n)=\bm a_0+\Omega J
    \begin{pmatrix}s\\ n\end{pmatrix}
    +\bm U_{\rm rem}(s,n),
    \label{eq:ridge-asymptotics}
\end{equation}
with no Williams $r^{-1/2}$ stress singularity in the leading class.
\end{proposition}

\begin{proof}
The blow-up analysis gives a transverse Neumann zero mode as the only leading finite mode.
The sectional balance shows that non-zero force and moment channels require a
finite terminal cut-off if stresses are bounded by a fixed material strength.
The scalar horn model gives the same admissibility condition for the averaged displacement.
The remaining leading modes are the rigid translations and rotation; strains enter through lower-order horn corrections and through $O(\eps_a^2)$ dynamic terms.
\end{proof}

\section{Cuspidal gorge as a zero-opening notch}

The local class of a gorge is different.
The removed cusp has half-width $b(s)=Bs^m$, and $b(s)/s\to0$.
Therefore the leading local geometry of the elastic body is a plane with a crack, not a horn.

\begin{assumption}
The outer Rayleigh-wave field is smooth on the inner scale and loads the gorge bottom through a bounded remote strain.
The local field belongs to the standard finite-energy crack-tip class for the static Lam\'e system.
\end{assumption}

Let $(r,\theta)$ be polar coordinates centered at the bottom of the gorge.
For plane strain set
\begin{equation}
    \kappa_\nu=3-4\nu .
    \label{eq:kappa-nu}
\end{equation}
Then the leading Williams expansion gives
\begin{equation}
    \bm U(r,\theta)
    =
    \bm a_0+\Omega J\bm X
    +
    \frac{K_I}{2\mu}\sqrt{\frac{r}{2\pi}}\Phi_I(\theta;\nu)
    +
    \frac{K_{II}}{2\mu}\sqrt{\frac{r}{2\pi}}\Phi_{II}(\theta;\nu)
    +
    O(r),
    \label{eq:gorge-displacement}
\end{equation}
and
\begin{equation}
    \sigma(r,\theta)
    =
    \frac{K_I}{\sqrt{2\pi r}}\Psi_I(\theta;\nu)
    +
    \frac{K_{II}}{\sqrt{2\pi r}}\Psi_{II}(\theta;\nu)
    +
    O(1).
    \label{eq:gorge-stress}
\end{equation}
The coefficients $K_I$ and $K_{II}$ are stress-intensity factors determined by matching to the outer Rayleigh-wave field.

The cusp shape changes the next terms.
Since the boundary displacement from a straight crack face is $b(r)=Br^m$, while the leading field varies on scale $r$, the relative correction is
\begin{equation}
    \frac{b(r)}{r}=Br^{m-1}.
    \label{eq:relative-cusp-correction}
\end{equation}
Therefore the first cusp correction to the crack-tip stress has the absolute order
\begin{equation}
    |\sigma_{\rm corr}|=O(r^{m-3/2}).
    \label{eq:cusp-correction}
\end{equation}
Figure~\ref{fig:local-scalings} summarizes the two local laws.

\begin{figure}[t]
    \centering
    \includegraphics[width=0.95\textwidth]{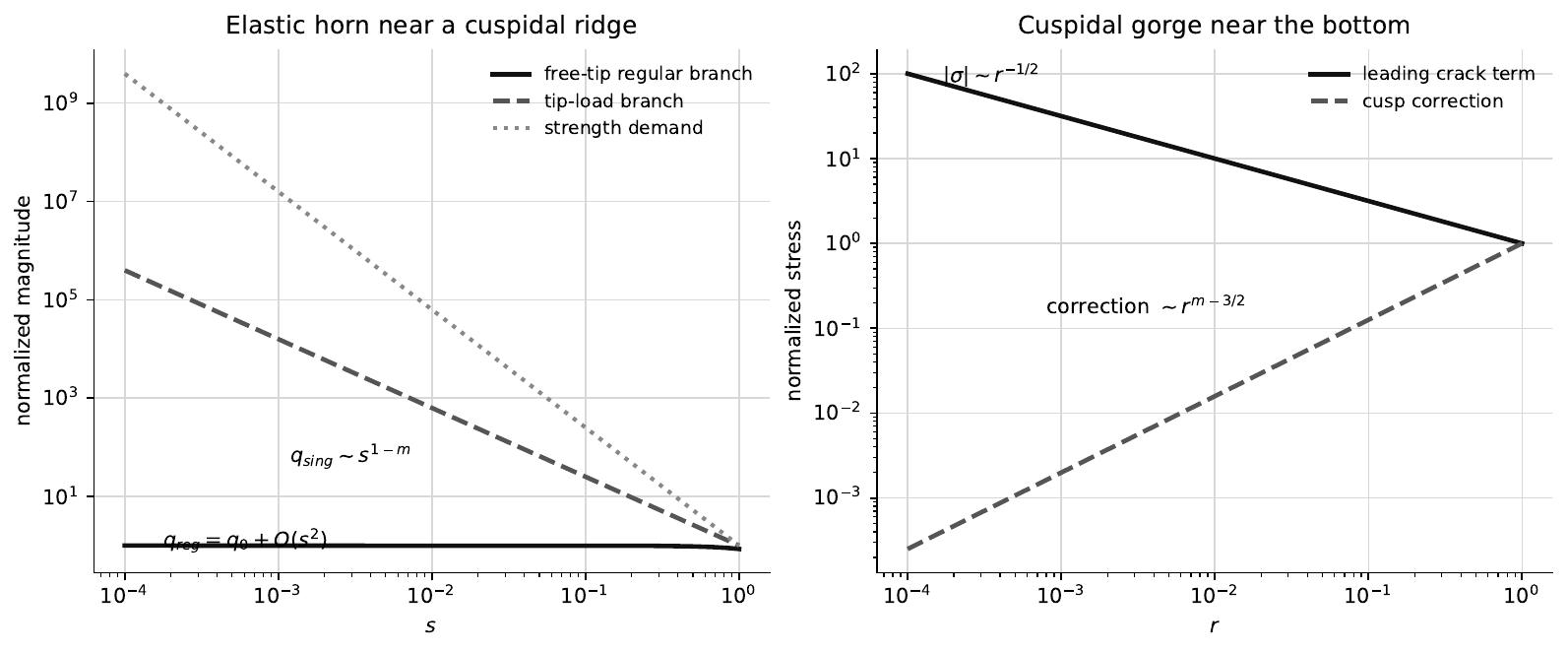}
    \caption{Model local laws. Left: in an elastic horn, a non-zero tip-load branch requires a finite cut-off if stresses are kept below a fixed material strength. Right: near the bottom of a cuspidal gorge the leading crack-tip stress behaves as $r^{-1/2}$, while the first cusp correction has order $r^{m-3/2}$.}
    \label{fig:local-scalings}
\end{figure}

\section{Matching to the outer Rayleigh-wave field}

We now embed the local results into the outer scattering problem.
The construction follows the standard rule of matched asymptotic expansions:
\[
    \text{outer field}+\text{inner field}-\text{common part}.
\]
This method is classical in boundary-layer and singular-perturbation theory \cite{Cole1968Perturbation,Ilin1992Matching,KevorkianCole1981}.

\begin{figure}[t]
    \centering
    \includegraphics[width=0.72\textwidth]{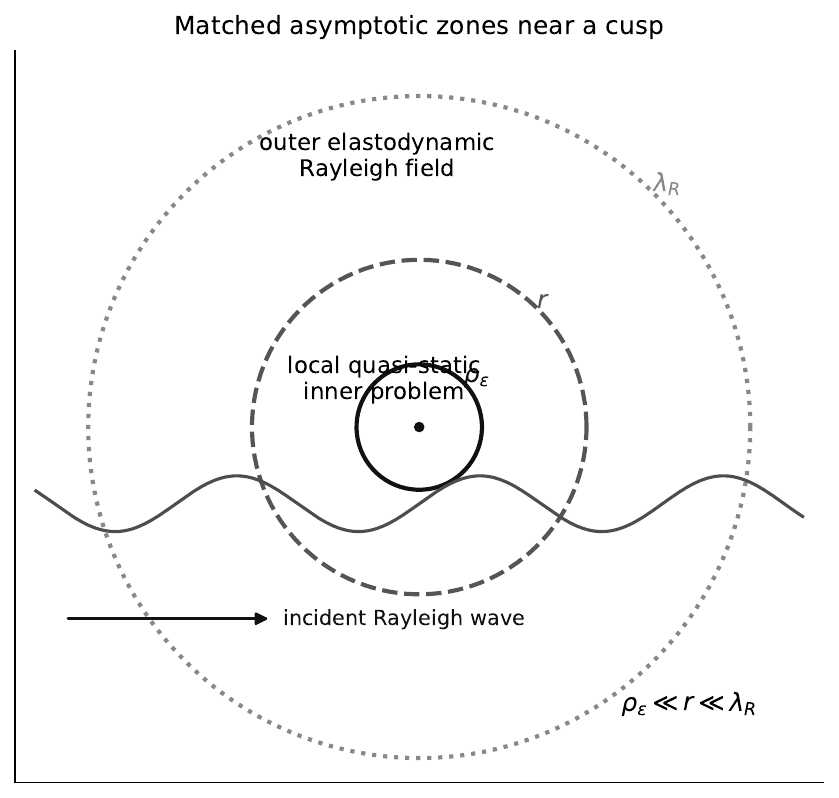}
    \caption{Matching structure. The outer wave region is governed by the Rayleigh-scale elastodynamic problem. The inner cusp region is quasi-static. In the overlap region the two expansions share a common part that must be subtracted in the composite approximation.}
    \label{fig:inner-outer-matching}
\end{figure}

The overlap region is non-empty, but its transverse scale depends on the local
geometry.  For a ridge we use the anisotropic horn variables
\[
    s=\ell_a r,\qquad n=B\ell_a^m r^m\eta,\qquad |\eta|\le1 .
\]
One possible overlap is
\begin{equation}
    1\ll r\ll \ell_a^{-1/2},\qquad |\eta|\le1 .
    \label{eq:matching-overlap-inner}
\end{equation}
Equivalently,
\begin{equation}
    \ell_a\ll s\ll \ell_a^{1/2},
    \qquad
    |n|\le Bs^m .
    \label{eq:matching-overlap-outer}
\end{equation}
In this region $a|s|^{\alpha-2}=r^{\alpha-2}\ll1$, so the outer expansion is
valid, while $s,n\to0$, so the outer field can still be Taylor-expanded at the
cusp.  The restriction $r\ll\ell_a^{-1/2}$ is only a convenient choice that
keeps the Taylor remainder smaller than the retained linear common part; many
intermediate choices give the same matching.

The normalized outer field restricted to the ridge horn is
\begin{equation}
    \mathcal U_{\rm out}^{(a),+}(r,\eta)
    =
    U_a^{-1}
    \left[
        U_{\rm out}(\ell_a r,B\ell_a^m r^m\eta)
        -
        U_{\rm rig}^{(a)}(\ell_a r,B\ell_a^m r^m\eta)
    \right].
    \label{eq:outer-field-inner-scale}
\end{equation}
In the overlap region it has the Taylor expansion
\begin{equation}
    \mathcal U_{\rm out}^{(a),+}(r,\eta)
    =
    A_{\rm out}^{(a)}
    +
    G_{\rm out}^{(a)}
    \begin{pmatrix}
        r\\
        B\ell_a^{m-1}r^m\eta
    \end{pmatrix}
    +
    O(\ell_a r^2).
    \label{eq:outer-taylor-inner-scale}
\end{equation}
Here
\begin{equation}
    G_{\rm out}^{(a)}
    =
    \Omega_{\rm out}^{(a)}J+E_{\rm out}^{(a)},
    \qquad
    E_{\rm out}^{(a)}=(E_{\rm out}^{(a)})^T .
    \label{eq:outer-gradient-split}
\end{equation}

For a ridge, the inner field has the far form
\begin{equation}
    \mathcal U_{\rm in}^{+}
    \sim
    A^+
    +
    \Omega^+J
    \begin{pmatrix}
        r\\
        B\ell_a^{m-1}r^m\eta
    \end{pmatrix}
    +
    \mathcal W_+^E(r,\eta;E^+)
    +
    \mathcal U_{\rm ev}^+,
    \label{eq:ridge-inner-far}
\end{equation}
where $\mathcal W_+^E$ is the regular horn corrector driven by the symmetric external strain and $\mathcal U_{\rm ev}^+$ contains exponentially small transverse modes.
Matching gives
\begin{equation}
    A^+=A_{\rm out}^{(a)},\qquad
    \Omega^+=\Omega_{\rm out}^{(a)},\qquad
    E^+=E_{\rm out}^{(a)}.
    \label{eq:ridge-matching-rigid}
\end{equation}
The common part is
\begin{equation}
    \mathcal U_{\rm com}^{+}(r,\eta)
    =
    A_{\rm out}^{(a)}
    +
    G_{\rm out}^{(a)}
    \begin{pmatrix}
        r\\
        B\ell_a^{m-1}r^m\eta
    \end{pmatrix}.
    \label{eq:ridge-common-part}
\end{equation}
No crack stress-intensity factor is produced by the leading ridge problem.

For a gorge the local limit is a crack, and the usual isotropic inner
variables $s=\ell_a x$, $n=\ell_a y$ are appropriate.  The overlap is
\[
    1\ll \sqrt{x^2+y^2}\ll\ell_a^{-1/2},
\]
and the normalized outer field has the same Taylor form as
\eqref{eq:outer-taylor-inner-scale} with
$(r,B\ell_a^{m-1}r^m\eta)$ replaced by $(x,y)$.
For a gorge, the inner expansion is
\begin{equation}
\begin{aligned}
    \mathcal U_{\rm in}^{-}
    &=
    A^{-}
    +
    \Omega^{-}J
    \begin{pmatrix}x\\y\end{pmatrix}
    +
    E^{-}
    \begin{pmatrix}x\\y\end{pmatrix}
    \\
    &\quad+
    \frac{\mathsf K_I}{2\mu}\sqrt{\frac{\varrho}{2\pi}}\Phi_I(\theta;\nu)
    +
    \frac{\mathsf K_{II}}{2\mu}\sqrt{\frac{\varrho}{2\pi}}\Phi_{II}(\theta;\nu)
    +
    \mathcal U_{\rm corr}^{-}.
\end{aligned}
    \label{eq:gorge-inner-expansion}
\end{equation}
The stress expansion contains
\begin{equation}
    \widehat\sigma_{\rm in}^{-}
    =
    \frac{\mathsf K_I}{\sqrt{2\pi\varrho}}\Psi_I
    +
    \frac{\mathsf K_{II}}{\sqrt{2\pi\varrho}}\Psi_{II}
    +
    O(1)+O(\varrho^{m-3/2}).
    \label{eq:gorge-inner-stress-expansion}
\end{equation}
The regular coefficients match as before, while the crack amplitudes are determined by a linear matching operator,
\begin{equation}
    \begin{pmatrix}
        \mathsf K_I\\
        \mathsf K_{II}
    \end{pmatrix}
    =
    \mathcal K_-\left(E_{\rm out}^{(a)},\lambda,\mu,\text{orientation}\right)
    +
    O(\eps_a^2).
    \label{eq:gorge-matching-K}
\end{equation}

The composite approximations are
\begin{equation}
    \mathcal U_{\rm unif}^{+}
    =
    \mathcal U_{\rm out}^{(a),+}
    +
    \mathcal U_{\rm in}^{+}
    -
    \mathcal U_{\rm com}^{+},
    \label{eq:ridge-composite}
\end{equation}
and
\begin{equation}
    \mathcal U_{\rm unif}^{-}
    =
    \mathcal U_{\rm out}^{(a),-}
    +
    \mathcal U_{\rm in}^{-}
    -
    \mathcal U_{\rm com}^{-}.
    \label{eq:gorge-composite}
\end{equation}
The same rule applies to stresses:
\begin{equation}
    \sigma_{\rm unif}^{\pm}
    =
    \sigma_{\rm out}^{\pm}
    +
    \sigma_{\rm in}^{\pm}
    -
    \sigma_{\rm com}^{\pm}.
    \label{eq:uniform-stress-composite}
\end{equation}

\section{Rounded cusp and stress cut-off}

For numerical and physical purposes the mathematical cusp is rounded.
Consider the regularized graph
\begin{equation}
    s_\delta(n)
    =
    A(n^2+\delta^2)^{\alpha/2}
    -
    A\delta^\alpha,
    \qquad 0<\alpha<1.
    \label{eq:regularized-graph}
\end{equation}
At $n=0$,
\begin{equation}
    s_\delta(n)
    =
    \frac{A\alpha}{2}\delta^{\alpha-2}n^2
    +
    O(\delta^{\alpha-4}n^4),
    \label{eq:regularized-expansion}
\end{equation}
and the local radius of curvature is
\begin{equation}
    \rho_\delta\sim\frac1{A\alpha}\delta^{2-\alpha}.
    \label{eq:rho-delta}
\end{equation}

\begin{corollary}[gorge cut-off]
If a cuspidal gorge has the crack-tip stress \eqref{eq:gorge-stress} and its bottom is rounded with curvature radius $\rho_\delta$, then
\begin{equation}
    \norm{\sigma}_{\max}
    \asymp
    \frac{|K|}{\sqrt{\rho_\delta}}
    \asymp
    C|K|\delta^{-(2-\alpha)/2},
    \qquad
    |K|=(K_I^2+K_{II}^2)^{1/2}.
    \label{eq:stress-cutoff}
\end{equation}
\end{corollary}

\begin{corollary}[ridge cut-off]
For the leading non-destructive free-tip ridge branch, rounding the tip does
not produce a universal crack-like law $\delta^{-1/2}$.
Any dependence on $\delta$ must come from lower-order geometry, global
focusing, resonance, or from a force/moment resultant that is physically
supported by a finite cut-off.
\end{corollary}

\begin{figure}[t]
    \centering
    \includegraphics[width=0.72\textwidth]{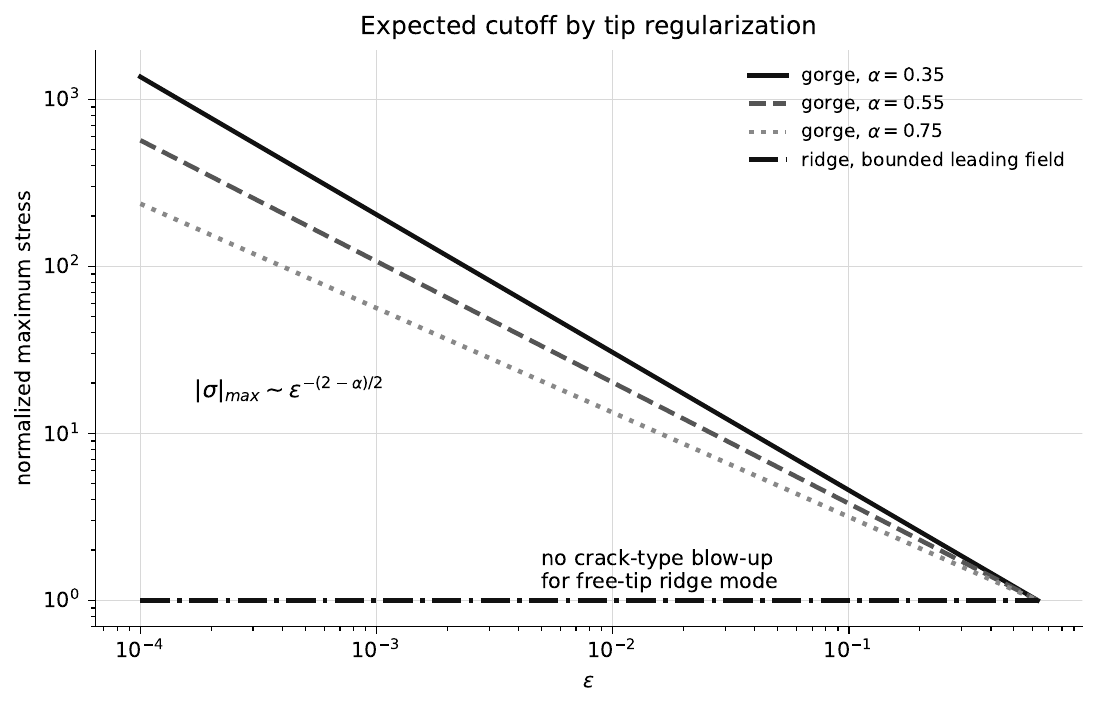}
    \caption{Expected rounding law. A cuspidal gorge inherits the crack cut-off $\|\sigma\|_{\max}\sim\delta^{-(2-\alpha)/2}$, whereas the leading non-destructive free-tip ridge branch remains non-crack-like.}
    \label{fig:regularization}
\end{figure}

\section{Numerical checks of first corrections}

Before solving a full finite-element problem, it is useful to check the reduced first-correction mechanisms numerically.
The first test concerns the matching variable.
For the modulated outer matching, the curvature term is controlled by
$as^{\alpha-2}$; after the substitution $s=\ell_a r$, this becomes the common
inner tail $r^{\alpha-2}$.
Figure~\ref{fig:outer-inner-first-correction} shows this inner transition
and compares the singular common tail with a bounded inner model correction
having the same far-field behaviour.
For the parameters used in the plot, the relative mismatch between the bounded
inner correction and its common tail is $1.29\times10^{-1}$ at $r=10$,
$4.64\times10^{-2}$ at $r=30$, and $1.43\times10^{-2}$ at $r=100$.

\begin{figure}[t]
    \centering
    \includegraphics[width=0.95\textwidth]{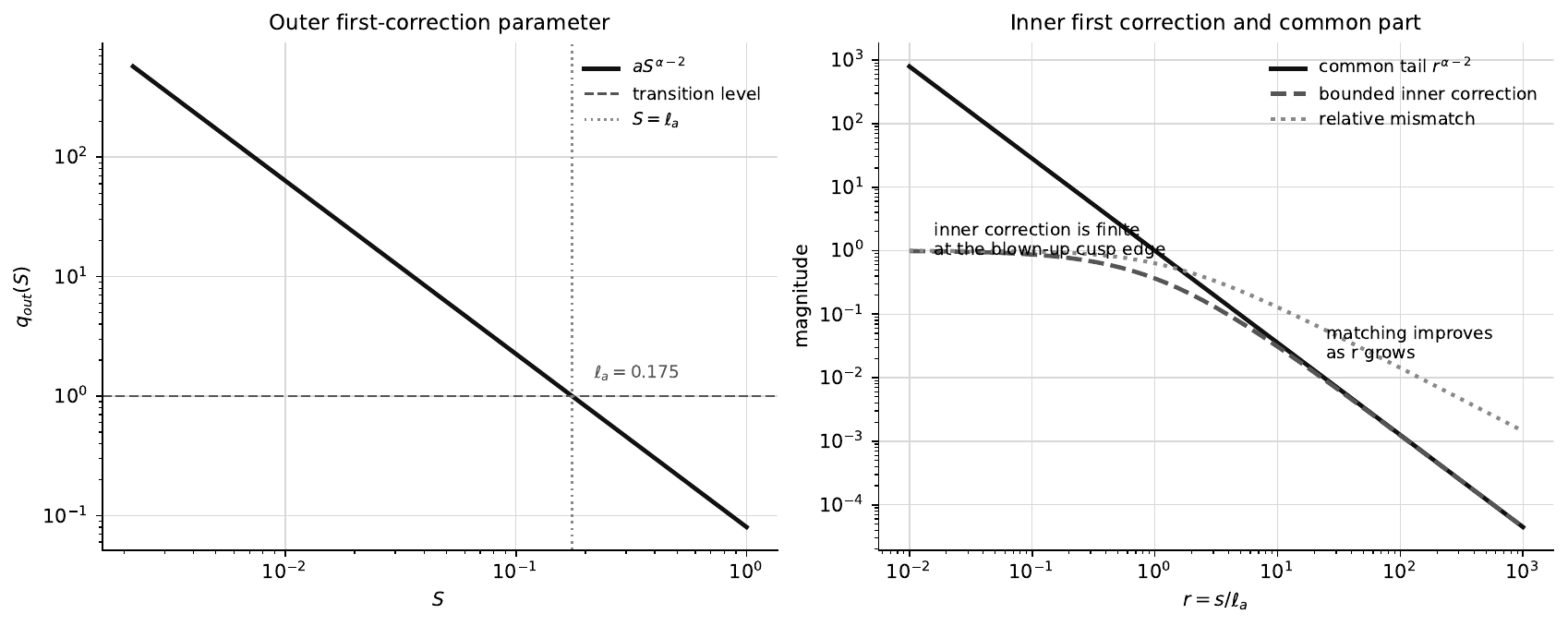}
    \caption{Numerical illustration of the first-correction matching. Left: the curvature parameter $as^{\alpha-2}$ reaches order one at the inner scale $s=\ell_a$. Right: in the longitudinal inner variable, the common tail $r^{\alpha-2}$ is recovered by a bounded inner correction as $r\to\infty$.}
    \label{fig:outer-inner-first-correction}
\end{figure}

The second test uses the scalar horn equation for the non-destructive free-tip ridge branch and the crack-tip model for the gorge.
For the horn equation, the numerical regular branch is compared with the first correction of the $\eps$-expansion; the maximum error scales as $\eps^{3.9968}$ in the tested range, consistent with the expected $O(\eps^4)$ remainder.
For the gorge model, the local log--log slope of the combined stress tends to $-0.4993$ at small radius, consistent with the leading $r^{-1/2}$ Williams term.
The corresponding plots are shown in Fig.~\ref{fig:inner-first-correction-models}.

\begin{figure}[t]
    \centering
    \includegraphics[width=0.95\textwidth]{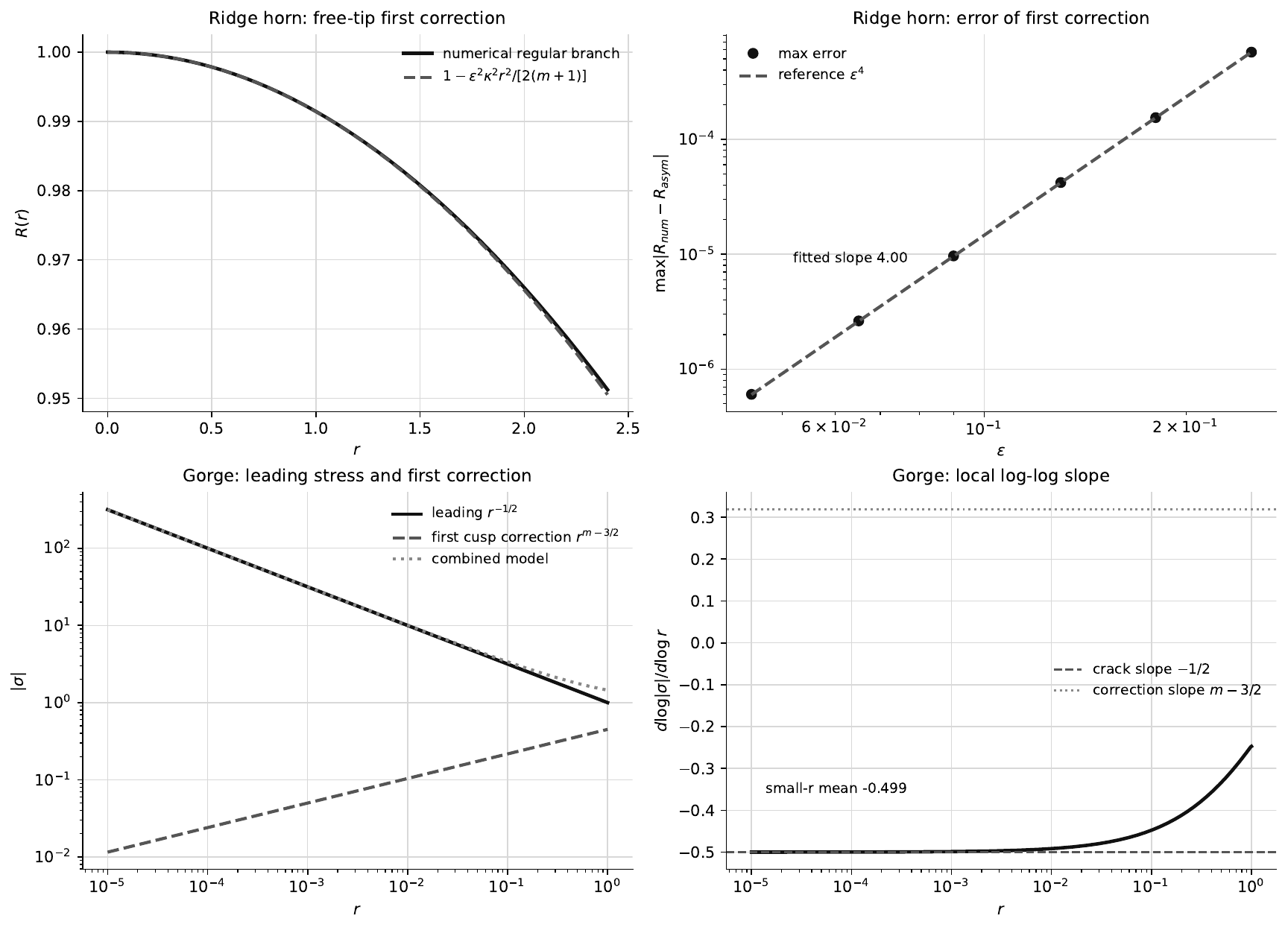}
    \caption{Reduced numerical checks of first corrections. Top: the regular free-tip horn branch agrees with its first $\eps^2$ correction and the error scales as $O(\eps^4)$. Bottom: the gorge stress model combines the crack-tip term with the first cusp correction; the small-radius local slope approaches $-1/2$.}
    \label{fig:inner-first-correction-models}
\end{figure}

\section{Finite-element verification of the local exponents}

The local predictions were checked with a simple two-dimensional finite-element computation for the static Lam\'e system.
Linear triangular elements were used, with $\lambda=\mu=1$ and $m=2.4$.
The ridge mesh is generated by mapping a rectangle in $(s,\eta)$ to $|n|<Bs^m$.
The gorge mesh is assembled from a lower block and two side blocks representing the material outside the cuspidal notch.
The cusp sides are traction-free in both tests.

For the ridge, the test is designed to excite the non-destructive free-tip branch rather than a force/moment channel supported by a finite terminal cut-off.
A distributed body load is applied while the remote section is fixed.
This choice is the reason why the high-energy branch is not seen in the
computed ridge profiles.  If $f_0$ is a bounded body force density, then the
resultant transmitted through the section $s$ is of order
\begin{equation}
    N_f(s)
    \asymp
    \int_0^s b(\xi)f_0\,\dd \xi
    =
    O(s^{m+1}).
    \label{eq:distributed-load-resultant}
\end{equation}
The corresponding mean stress is $N_f(s)/A(s)=O(s)$ and therefore vanishes at
the tip.  Thus the numerical ridge test approximates the zero-tip-load,
bounded-stress limit.
For a truncated ridge with lower cut-off $s=\rho$, the $95$th percentile of $|\sigma|$ is measured in the zone $2\rho<s<4\rho$.
The fitted log--log slope is approximately $+1.02$, so the stress does not grow as $\rho\to0$ in this free-tip test.

A second ridge computation was made for the general high-energy branch of
Proposition~\ref{prop:high-energy-tip-load-branch}.
The same horn was truncated at $s=\rho$, the remote section was fixed, and a
fixed non-zero resultant force was distributed over the terminal section
$s=\rho$.
This is a different boundary-value problem: the total force is kept fixed
while the available cross-sectional area is $A(\rho)=2B\rho^m$.
The predicted stress and energy laws are therefore
\[
    |\sigma|_{\rm tip}\asymp \rho^{-m},
    \qquad
    E(\rho)\asymp \rho^{1-m}.
\]
For $m=2.4$, the finite-element slopes are $-2.46$ for the near-tip stress
and $-1.41$ for the total elastic energy, close to the theoretical values
$-2.4$ and $-1.4$.
The same code was then rerun for $m=1.8$ and $m=2.8$.
For the forced ridge branch the measured stress slopes are $-1.85$ and $-2.87$,
while the energy slopes are $-0.87$ and $-1.80$; these values track the laws
$-m$ and $1-m$ in the expected way.
For the gorge the fitted slope remains essentially unchanged, between $-0.495$
and $-0.494$, confirming that the leading crack exponent does not depend on $m$.
The comparison is shown in Fig.~\ref{fig:fem-tip-load-branch}.

For the gorge, the outer boundary is loaded by a remote horizontal strain that opens the notch.
The $95$th percentile of $|\sigma|$ is measured in rings around the bottom.
The fitted slope in the intermediate range is $-0.495$, in close agreement with the crack-tip law $r^{-1/2}$.
The comparison is summarized in Table~\ref{tab:fem-comparison}.

\begin{table}[t]
    \centering
    \caption{Comparison of local asymptotics and FEM observations.}
    \label{tab:fem-comparison}
    \begin{tabular}{@{}p{0.25\textwidth}p{0.39\textwidth}p{0.26\textwidth}@{}}
        \toprule
        Geometry & Asymptotic prediction & FEM observation\\
        \midrule
        Free-tip ridge & no universal $r^{-1/2}$ growth & cut-off slope $+1.02$\\
        Ridge with fixed tip force & $|\sigma|\asymp \rho^{-m}$, $E\asymp\rho^{1-m}$ & slopes $-2.46$, $-1.41$\\
        Gorge & $|\sigma|\asymp r^{-1/2}$ & profile slope $-0.495$\\
        \bottomrule
    \end{tabular}
\end{table}

\begin{figure}[t]
    \centering
    \includegraphics[width=0.95\textwidth]{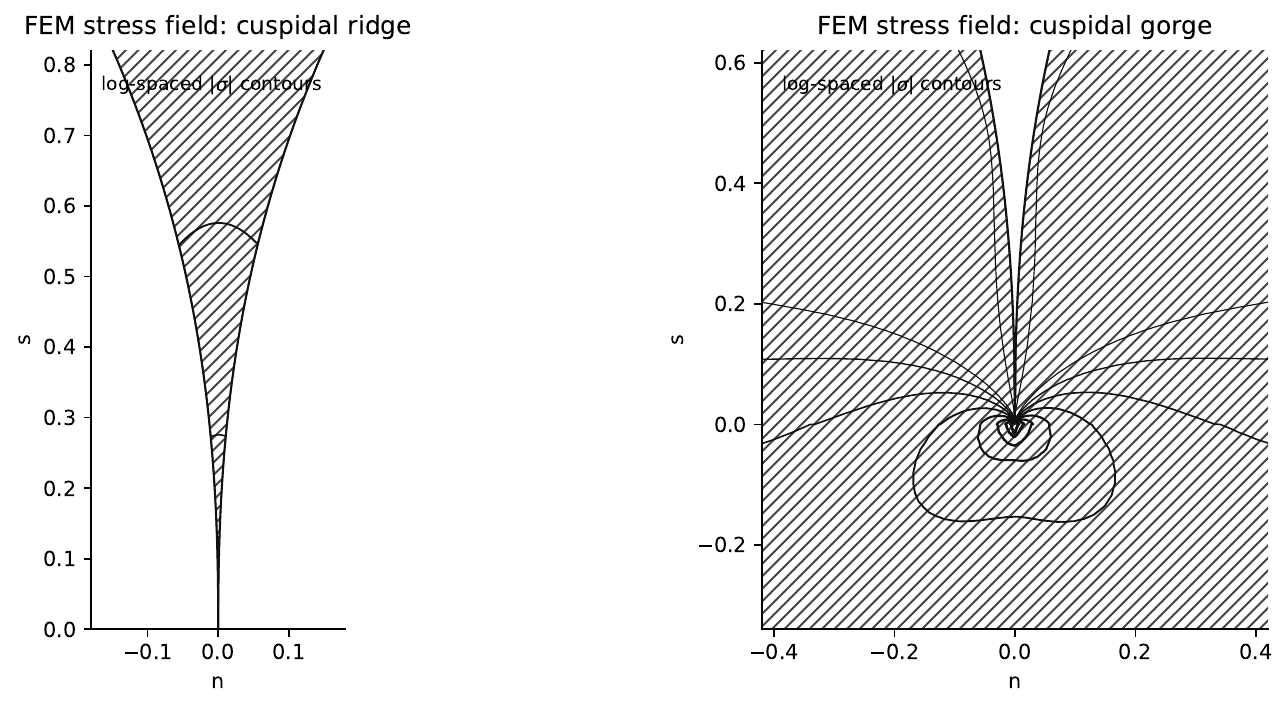}
    \caption{Local FEM stress fields for the cuspidal ridge and gorge. The material region is hatched, and contour lines show logarithmic levels of $|\sigma|$.}
    \label{fig:fem-stress-fields}
\end{figure}

\begin{figure}[t]
    \centering
    \includegraphics[width=0.95\textwidth]{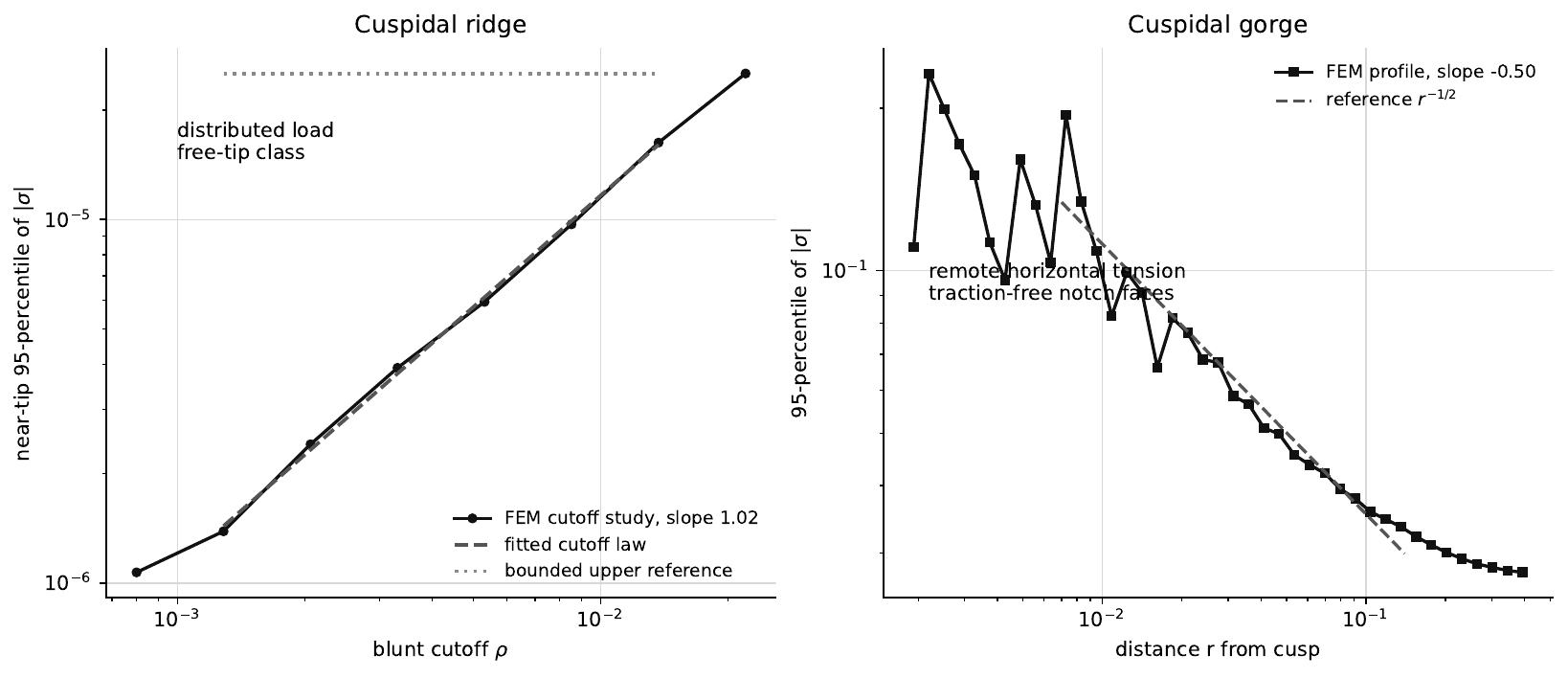}
    \caption{Numerical scaling check. The non-destructive free-tip ridge test does not show crack-like growth as the cut-off decreases. The gorge profile has slope close to $-0.5$.}
    \label{fig:fem-scaling}
\end{figure}

\begin{figure}[t]
    \centering
    \includegraphics[width=0.95\textwidth]{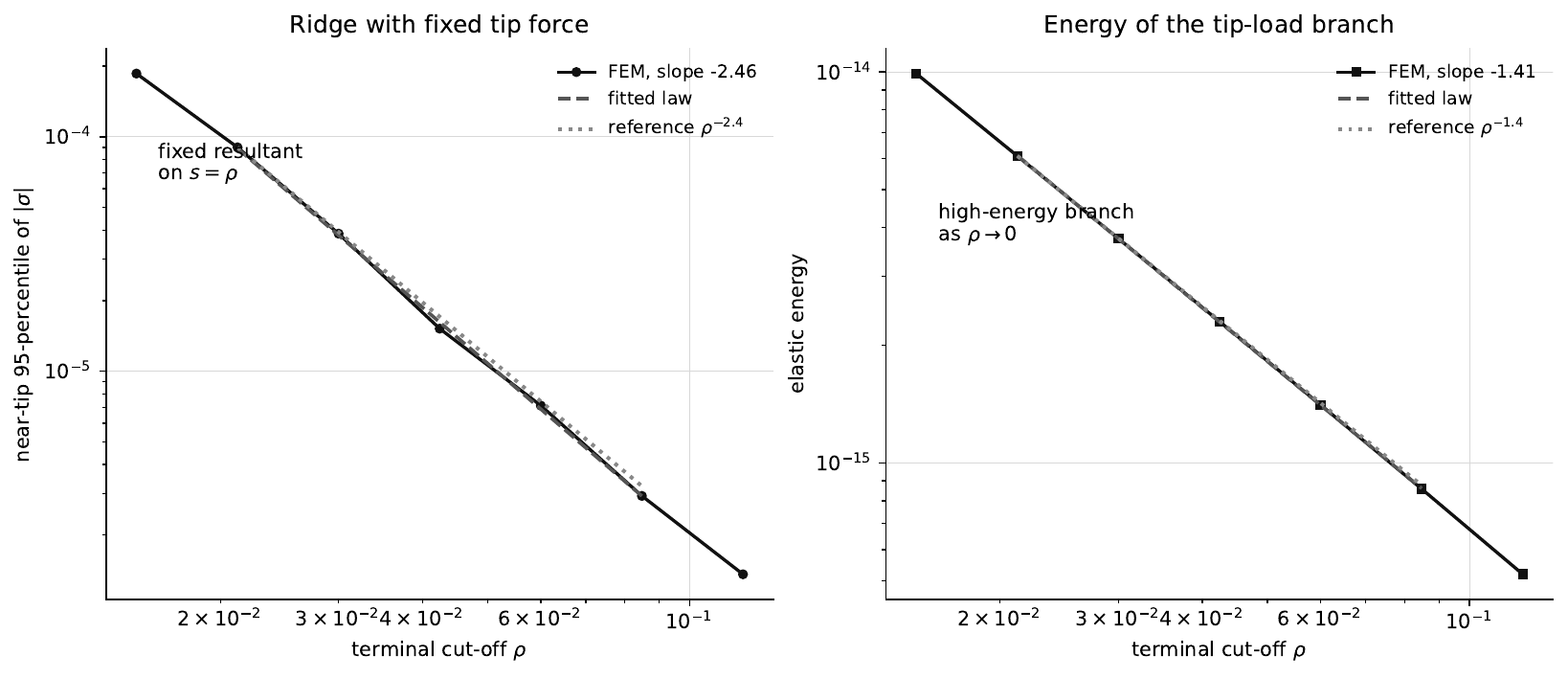}
    \caption{Numerical scaling check for the general high-energy ridge branch. A fixed non-zero resultant is applied to the terminal section $s=\rho$. The stress slope is close to $-m$, and the elastic-energy slope is close to $1-m$.}
    \label{fig:fem-tip-load-branch}
\end{figure}

These calculations are not full simulations of Rayleigh-wave scattering by a finite mountain shape.
They check the local elliptic mechanism that must be matched to the outer wave field.
For the ridge, arbitrary non-rigid boundary data can excite a force or moment
channel that is physically supported only by a finite terminal cut-off.  Such
growth should be interpreted as a boundary-loading artefact or as a different
truncated-tip problem, not as a universal crack singularity of the free ridge.

\section{Scope and diagnostics}

The analysis is local, two-dimensional and linearly elastic; it does not
include attenuation, plasticity, fracture growth, gravity, layering, finite
amplitudes, or the global radiation problem for a finite topographic object.
Its output is the local asymptotic structure and the matching data needed by
such a global problem.

The classification gives practical diagnostics for wave solvers.  For a
cuspidal ridge, mesh refinement should not reveal a universal crack-like law
in the leading non-destructive free-tip class; if a robust $\delta^{-1/2}$ law
appears, concentrated loads, artificial constraints, finite tip truncation,
geometry errors, or global resonance should be checked.  For a cuspidal gorge,
the crack law is expected until it is cut off by rounding, and the coefficients
$K_I,K_{II}$ describe how the outer Rayleigh field loads the notch bottom.

The least closed part is the fully vectorial ridge result.
The sectional resultant balance and the scalar horn model rigorously identify
the force and moment channels that require a finite terminal cut-off under a
fixed material strength.
The statement that the complete leading vector field is asymptotically rigid is used as a formal horn asymptotic result supported by Korn-type reasoning.
A fully rigorous theorem would require weighted Korn inequalities and a complete asymptotic decomposition for the two-dimensional Lam\'e system in a cusp horn.
This is a different mathematical problem from the applied matched-asymptotic construction pursued here.

\section{Conclusions}

The local asymptotic analysis gives the following conclusions.

\begin{enumerate}
    \item The Rayleigh wavelength defines the outer wave scale.  Away from the cusp the surface is nearly flat for small $a=A\lambda_R^{\alpha-1}$, but the effective Rayleigh wavelength parameter must be corrected order by order.
    \item The modulated Rayleigh-scale expansion first loses uniformity when $a|s|^{\alpha-2}=O(1)$, giving the longitudinal inner scale $\ell_a=a^{1/(2-\alpha)}$.
    \item For a ridge, the matching inner variables are anisotropic: $s=\ell_a r$ and $n=B\ell_a^m r^m\eta$.  On this scale the leading equation is static; inertia is an $O(\eps_a^2)$ correction with $\eps_a=2\pi\ell_a$.
    \item A cuspidal ridge is a vanishing-width elastic horn.  With a fixed admissible stress level, non-zero force and moment resultants can be carried only by a horn with a finite terminal cut-off.
    \item The leading non-destructive free-tip ridge field is asymptotically rigid and does not contain a Williams $r^{-1/2}$ stress singularity.
    \item A cuspidal gorge is a zero-opening re-entrant notch.  Its leading field is the crack-tip Williams field with $|\sigma|\asymp r^{-1/2}$.
    \item The cusp exponent enters the gorge through the correction $O(r^{m-3/2})$ and through the rounding law $\|\sigma\|_{\max}\asymp\delta^{-(2-\alpha)/2}$.
    \item The stress-intensity factors are not local geometric constants; they are determined by matching to the outer Rayleigh-wave field.
    \item The composite approximation is obtained from the rule $U_{\rm out}+U_{\rm in}-U_{\rm com}$.
    \item Local FEM calculations support the predicted ridge--gorge distinction.
\end{enumerate}

\section*{Declaration of generative AI and AI-assisted technologies in the writing process}

During the preparation of this manuscript, generative AI-assisted writing and coding tools were used for language editing, structuring assistance, code commenting, and symbolic/numerical consistency checks. The author critically reviewed and approved all content, including mathematical derivations, interpretations, code, numerical results, and references.
The AI tools were not used as a source of primary scientific data and are not authors of the work.
The author is fully responsible for the content, correctness and final form of the manuscript.

\end{document}